\documentclass[12pt,english]{article}
\usepackage[a4paper, portrait,
	left=0.75in,
	right=0.75in,
	top=1in,
	bottom=0.5in,
	footskip=.25in]{geometry}

\usepackage{mathrsfs}
\usepackage{amsfonts}
\usepackage{amsmath} 
\usepackage{amssymb} 
\usepackage{amsthm}

\usepackage{mathrsfs}
\usepackage{physics}

\usepackage{graphicx}
\usepackage{array}
\usepackage{xspace}
\usepackage{natbib}

\usepackage{orcidlink}
\usepackage{authblk}

\usepackage{url}
\usepackage{hyperref}
\usepackage{color}
\definecolor{refcolor}{RGB}{160,35,0}
\definecolor{hrefcolor}{RGB}{0,35,190}
\hypersetup{
    colorlinks,
    citecolor=refcolor,
    filecolor=refcolor,
    linkcolor=hrefcolor,
    urlcolor=hrefcolor
}

\newtheorem{theorem}{Theorem}
\newtheorem{question}{Question}
\newtheorem{insight}{Insight}

\theoremstyle{remark}
\theoremstyle{definition}

\newtheorem{proposition}{Proposition}
\newtheorem{definition}{Definition}

\newtheorem{remark}{Remark}

\theoremstyle{definition}
\newtheorem{note}{Note}

\newtheoremstyle{nospace} 
  {10pt}                   
  {10pt}                   
  {\itshape}              
  {}                      
  {\bfseries}             
  {.}                     
  {.5em}                  
  {\thmname{#1}\thmnumber{#2}\thmnote{ (#3)}} 

\theoremstyle{nospace}
\newtheorem{QT}{QT}
\newcommand{\qtref}[1]{{\normalfont{QT\ref{#1}}}}

\newcommand{\HSF}{\ref{def:HSF}\xspace}
\newcommand{\thmStoica}{\ref{thm:theoremtwo}\xspace}

\newtheorem{defCustom}{}
\makeatletter
\newcommand{\setdefCustomtag}[1]{
  \let\oldthedefCustom\thedefCustom
  \renewcommand{\thedefCustom}{{\normalfont\textbf{#1}}}
  \g@addto@macro\enddefCustom{
    \global\let\thedefCustom\oldthedefCustom}
  }
\makeatother

\newcommand{\orcid}[1]{\href{https://orcid.org/#1}{\textcolor[HTML]{A6CE39}{\aiOrcid}}}

\def\({\left(}
\def\){\right)}
\def\[{\left[}
\def\]{\right]}

\newcommand{\hilbert}{\mathcal{H}}
\newcommand{\mc}[1]{\mathcal{#1}}
\newcommand{\ms}[1]{\mathscr{#1}}

\newcommand{\wt}[1]{\widetilde{#1}}

\newcommand{\pobs}[1]{#1}
\newcommand{\obs}[1]{\mathsf{\pobs{#1}}}
\newcommand{\oper}[1]{\obs{#1}}

\newcommand{\R}{\mathbb{R}}
\newcommand{\C}{\mathbb{C}}

\newcommand{\ie}{\textit{i.e.}\xspace}
\newcommand{\vs}{\textit{vs.}\xspace}
\newcommand{\eg}{\textit{eg.}\xspace}

\newcommand{\etc}{\textit{etc}\xspace}

\newcommand{\Stab}{\textnormal{Stab}}

\newcommand{\U}{\textnormal{U}}

\newcommand{\SU}{\textnormal{SU}}

\newcommand{\x}{\mathbf{x}}

\newcommand{\q}{\boldsymbol{q}}
\newcommand{\p}{\boldsymbol{p}}

\def\sref #1{\S\ref{#1}}

\author{Cristi Stoica\ \orcidlink{0000-0002-2765-1562}}

\affil{\small Dept. of Theoretical Physics, NIPNE-HH, Bucharest, Romania.\\\href{mailto:cristi.stoica@theory.nipne.ro}{cristi.stoica@theory.nipne.ro},  \href{mailto:holotronix@gmail.com}{holotronix@gmail.com}}

\title{Comment on ``On the emergence of preferred structures in quantum theory'' by Soulas, Franzmann, and Di Biagio}

\date{\small\today} 

\begin{document}

\pagestyle{headings}	
\newpage
\setcounter{page}{1}
\renewcommand{\thepage}{\arabic{page}}

\maketitle

\begin{abstract}
This reply is also a friendly introduction to the impossibility of emergence of preferred structures from the Hamiltonian $\mathsf{H}$ and the unit vector $|\psi\rangle$ only. The obstructions to emergence are illustrated on the concrete construction of a tensor product structure (TPS) from \citealp{SoulasFranzmannDiBiagio2025OnTheEmergenceOfPreferredStructuresInQuantumTheoryV2} (arXiv:2512.07468v2). Soulas et al. offer their TPS as a counterexample to the proof from \citealp{Stoica2022SpaceThePreferredBasisCannotUniquelyEmergeFromTheQuantumStructure} (arXiv:2102.08620) that structures constructed only from $\mathsf{H}$ and $|\psi\rangle$ either contradict physical observations or can't describe them unambiguously.

Soulas et al.'s construction of a unique TPS can't be both invariant and compatible with physical observations, so it can't be a counterexample. Its incompatibility becomes visible by examining how the relation between $|\psi(t)\rangle$ and the TPS, encoding the entanglement, changes in time. Therefore their TPS doesn't refute, but confirms \citep{Stoica2022SpaceThePreferredBasisCannotUniquelyEmergeFromTheQuantumStructure}.

Besides this, since Soulas et al.'s method to construct preferred structures consists of choosing their invariants, by the same logic one could claim as well that the masses of elementary particles emerge uniquely just by fixing their values by hand.

Soulas et al.'s construction is concrete and can illustrate the major obstructions for emergent structures, confirming them despite doing the best possible to avoid them. This makes it an excellent pedagogical tool to illustrate the trilemma, but also the relational and structural aspects of quantum theory and its symmetries.
\end{abstract}


\tableofcontents

\section{Introduction}
\label{s:intro}

This is a reply to \citep[arXiv:2512.07468v2]{SoulasFranzmannDiBiagio2025OnTheEmergenceOfPreferredStructuresInQuantumTheoryV2}, but also a friendly introduction to the various obstructions for the emergence of preferred structures from the Hamiltonian $\mathsf{H}$ and the unit vector $|\psi\rangle$ only, a hot research topic these days, spreading across multiple research programs, from foundations of quantum theory to quantum gravity.
I think the concrete construction of an emergent tensor product structure (TPS) proposed by Soulas et al. is excellent as an illustration of these obstructions.

Soulas et al. constructed a TPS and proposed it as a counterexample to the no-go results from \citep{Stoica2022SpaceThePreferredBasisCannotUniquelyEmergeFromTheQuantumStructure,Stoica2024DoesTheHamiltonianDetermineTheTPSAndThe3dSpace} against Hilbert space fundamentalism (\HSF).
\HSF is the thesis that a complete description of the world, including space, fields or particles, and the TPS giving the decomposition into subsystems, emerge uniquely from $\mathsf{H}$ and $|\psi\rangle$ only. In \citep{Stoica2022SpaceThePreferredBasisCannotUniquelyEmergeFromTheQuantumStructure} I showed that such constructions either can't describe physical changes in the world, or can't describe unambiguously the physical reality they're supposed to describe.

To select a physically unique TPS, one has to fix a number of parameters exponential in the number of subsystems, as proved in \citep{Stoica2024DoesTheHamiltonianDetermineTheTPSAndThe3dSpace}. Soulas et al.'s construction does this by fixing a continuum infinity of interdependent parameters representing entanglement invariants.

The central motivation of this article is to explain the trilemma faced by Soulas et al.'s construction, when confronted with time evolution:
\begin{enumerate}
	\item 
either its infinitely many invariants depend of time in a strictly choreographed way,
	\item 
or they have to be fixed at an absolute moment of time once and for all,
	\item 
or the solution contradicts the empirical observation that entanglement can change.
\end{enumerate}

The first two options violate the invariance of the construction and are equivalent to fixing an absolute TPS, as normally done in quantum theory, but by more complicated means that may give the impression of emergence.

Therefore, their construction gives a nice confirmation of the very no-go results from \citep{Stoica2022SpaceThePreferredBasisCannotUniquelyEmergeFromTheQuantumStructure,Stoica2024DoesTheHamiltonianDetermineTheTPSAndThe3dSpace} that it was hoped to refute, and provides an excellent illustration of the trilemma.

In addition to this, if we could claim that fixing the invariants that classify the solutions counts as unique emergence, we could equally claim that the masses and spins of the elementary particles from the standard model emerge uniquely by manually fixing the Casimir invariants of the irreducible representations of the Poincar\'e group, claim that nobody makes, for the obvious reason that these Casimir invariants are exactly the mass and spin.

Soulas et al. claim to identify errors and limitations in \citep{Stoica2022SpaceThePreferredBasisCannotUniquelyEmergeFromTheQuantumStructure}, to correct them, and to generalize the framework. A comparison with the original source shows that the alleged errors and limitations are not present in it.

Therefore, \HSF is still untenable. However, I explore the possibility that a weaker version of \HSF can survive by dramatically reducing its scope from a theory of everything to limited relational aspects of the physical laws. I illustrate this possibility by comparing what was really achieved about locality in \citep{CotlerEtAl2019LocalityFromSpectrum} with what was believed to be achieved.

I also include an analysis of the ``relational philosophy'' of quantum theory, brought by Soulas et al. as an argument for their construction. We will see whether physical evidence supports the relational perspective from \citep{SoulasFranzmannDiBiagio2025OnTheEmergenceOfPreferredStructuresInQuantumTheoryV2} or the one from \citep{Stoica2022SpaceThePreferredBasisCannotUniquelyEmergeFromTheQuantumStructure}.

\section{Hilbert space fundamentalism and emergent structures}
\label{s:HSF}

The Hilbert space fundamentalism program can be understood in several ways,
\begin{enumerate}
	\item 
As a theory of quantum gravity \citep{Piazza2010GlimmersOfAPreGeometricPerspective,Giddings2019QuantumFirstGravity} or even as a theory of everything  \citep{CarrollSingh2019MadDogEverettianism,Carroll2021RealityAsAVectorInHilbertSpace}, aiming to derive all facts about the world only from the Hamiltonian and the state vector.
	\item 
As a reconstruction of quantum theory, a derivation of some of the principles of quantum theory from others, similar to how mathematicians tried for centuries to deduce the parallel postulate from the other postulates of Euclidean geometry.
	\item 
As a framework to solve specific problems like the emergence of a pointer basis in decoherence \citep{JoosZeh1985EmergenceOfClassicalPropertiesThroughInteractionWithTheEnvironment,Zurek1998DecoherenceEinselectionAndTheExistentialInterpretation,Zurek2003DecoherenceEinselectionAndTheQuantumOriginsOfTheClassical,LombardiEtAl2012TheProblemOfIdentifyingSystemEnvironmentDecoherence}, emergence of space or spacetime in quantum gravity \citep{CarrollSingh2019MadDogEverettianism,Carroll2021RealityAsAVectorInHilbertSpace} \etc.
	\item 
As a natural consequence of
\begin{enumerate}
	\item 
 the focus of modern physics on the mathematical structures, in particular on dualities, multiple structures realizing the same physics, 
	\item 
a relational philosophy and views informed by quantum information theory,
	\item 
structural realism, which finds unnecessary and even misguided to talk about physical meaning beyond structures and mathematical relations,
	\item 
operationalism, instrumentalism.
\end{enumerate}
	\item 
	All of the above.
\end{enumerate}

\subsection{HSF as a project to simplify the foundations of quantum theory}
\label{s:intro:std-qt}

The standard specification of a quantum theory consists of the following
\begin{QT}
\label{QT:HSF}
The total state of the world is represented by a unit vector $\ket{\psi(t)}\in\hilbert$ from a Hilbert space $\hilbert$. Its time evolution is determined by a Hermitian operator $\oper{H}$ on $\hilbert$ (the Hamiltonian),
\begin{equation}
\label{eq:time-evol}
\ket{\psi(t)}=e^{-\frac{i}{\hbar}\oper{H}t}\ket{\psi(0)}.
\end{equation}
\end{QT}

\begin{QT}
\label{QT:TPS}
A decomposition of $\hilbert$ as a tensor product
\begin{equation}
\label{eq:hilbert-tensor}
\hilbert \cong \hilbert_1\otimes\hilbert_2\otimes\ldots,
\end{equation}
which describes the structure of the system in terms of subsystems or regions of space.
\end{QT}
\begin{QT}
\label{QT:OBS}
An association of physically observable properties to Hermitian operators on $\hilbert$ or on the subsystem spaces $\hilbert_j$. These are of the form

\begin{center}
\begin{tabular}{p{5cm} l l}
  position~of~particle~$j$ & $\mapsto$ & $(\widehat{x}_j,\widehat{y}_j,\widehat{z}_j)$ \\
  momentum~of~particle~$j$ & $\mapsto$ & $(\widehat{p}_{x_j},\widehat{p}_{y_j}, \widehat{p}_{z_j})$ \\
  $\ldots$ &  & $\ldots$ \\
\end{tabular}
\end{center}
or similar prescriptions in the case of quantum field operators.
\end{QT}

\begin{QT}
\label{QT:PROJ}
When measuring an observable $\oper{A}=\sum_a a \dyad{a}$ we obtain an eigenvalue $a$ of $\oper{A}$ with probability $\abs{\braket{a}{\psi}}^2$; then the system (or our description of the system) updates from $\ket{\psi}$ to $\ket{a}$.
\end{QT}

There is much more to say about this, but I try to ignore the details that, while essential for other research questions, are irrelevant for this discussion.

The Hilbert space as in \qtref{QT:HSF} already comes equipped with all Hermitian and unitary operators, just like the Euclidean space comes equipped with all isometries.
The role of \qtref{QT:OBS} is not to add these operators to $\hilbert$, since they're already there, but to prescribe their physical meaning.
Which operators represent the position of each particle, which operators represent their momenta and so on.
For example, the position operators $\widehat{\x}$ turn the unit vector $\ket{\psi(t)}$, which by itself is like any other unit vector from the Hilbert space, into something full of structure and physical meaning, the \emph{wavefunction}:
\begin{equation}
\label{eq:wavefunction}
\psi(\x,t)=\braket{\x}{\psi(t)}.
\end{equation}

Without this, it would be meaningless to talk about position observables or position basis or configuration space of positions, and therefore of a  wavefunction. Even in quantum field theory, which comes equipped with a 3D space, we need to know the domains in space and the physical meanings of the local operators, and in some formulations we need to know the configuration space of classical fields \citep{Hatfield2018QuantumFieldTheoryOfPointParticlesAndStrings}.

\qtref{QT:OBS} includes, among its association of physical meanings to observables, the observables generating \emph{local algebras of observables}, \ie what observables correspond to each subsystem.
This implies that \qtref{QT:OBS} determines the tensor product decomposition~\eqref{eq:hilbert-tensor} \citep{ZanardiLidarLloyd2004QuantumTensorProductStructuresAreObservableInduced}.
In other words, if \qtref{QT:OBS} is sufficiently specified, \qtref{QT:TPS} is not an independent axiom, and the tensor product decomposition emerges from \qtref{QT:OBS}.

In addition, the many-worlds interpretation proposes that \qtref{QT:PROJ} also emerges from \qtref{QT:HSF}, \qtref{QT:TPS}, and \qtref{QT:OBS} \citep{Everett1957RelativeStateFormulationOfQuantumMechanics,Wallace2012TheEmergentMultiverseQuantumTheoryEverettInterpretation,SEP-Vaidman2021MWI}.
In other words, the measurements, including the probabilities, and the emergence of the classically looking macroscopic level of reality, are consequences of unitary evolution due to decoherence. The emergence of probabilities and their meaning in the many-worlds interpretation are disputed, and there are proposals to address them, but this may need that some observables have a preferred ontological status while everything remaining consistent with \qtref{QT:HSF}--\qtref{QT:OBS} \citep{Stoica2023TheRelationWavefunction3DSpaceMWILocalBeablesProbabilities,Stoica2024ClassicalManyWorldsInterpretation}.
But the question whether the postulate \qtref{QT:PROJ} is redundant or not is not the subject of this article.

Here we discuss only the possibility to reconstruct \qtref{QT:TPS} and \qtref{QT:OBS} just from \qtref{QT:HSF}.

\subsection{Motivation for Hilbert space fundamentalism}
\label{s:intro:HSF}

Are postulates \qtref{QT:TPS} and \qtref{QT:OBS} redundant, or can they be reduced to \qtref{QT:HSF}? Is it necessary to postulate associated physical meanings to the operators, as in \qtref{QT:OBS}? It makes perfect sense to think that this meaning should follow straight from the spectral properties of the operators and from the relations of these operators with each other and with the Hamiltonian.

Shouldn't everything there is to know about the world be encoded in the relations between the measured properties and the objects used as measurement units with which we compare them, in the data resulting from experiments (the ``clicks of the detectors''), and in the correlations between these data?
After all, everything that we can access in the world we decode from this raw data, and based on the correlations between the raw data we reconstruct other structures like space, and we identify subsystems and their algebras of operators, recovering \qtref{QT:TPS} and all structures in the world.
We always did this even before the discovery of quantum theory, even before the invention of language, as mere biological beings: we were somehow able to turn the sensory input into virtual internal worlds, in our attempts to depict to ourselves the external world and its place in it.

Isn't everything that we can access just information?
Isn't everything describable mathematically, as sets and relations \citep{Hodges1997ShorterModelTheory,Gratzer2008UniversalAlgebra,Stoica2016AndTheMathWillSetYouFreeSpringer}? Or maybe they even \emph{are} just pure mathematical structures \citep{Tegmark2014OurMathematicalUniverse}.
Maybe to us, humans, it feels different, it feels that the field of experience is constituted of a wide range of colorful experiences to which we can't stay neutral, but what else do we experience, if not quantum information events, and from these bits we infer everything there is \citep{Wheeler1990InformationPhysicsQuantumLinks,FQXI2015ItFromBitOrBitFromIt}?

This is why to many it seems more than plausible that the only fundamental postulate needed by a quantum theory is \qtref{QT:HSF} (with or without \qtref{QT:PROJ}), and from \qtref{QT:HSF} alone we can recover both \qtref{QT:TPS} and \qtref{QT:OBS},
\begin{equation}
\label{eq:HSF-QT}
\text{\qtref{QT:HSF}} \Rightarrow \text{\qtref{QT:TPS}$+$\qtref{QT:OBS}}.
\end{equation}

It seems that we are unavoidably driven by logic itself to the minimalist ontology called \emph{Hilbert space fundamentalism} (\HSF), described in \citep{CarrollSingh2019MadDogEverettianism} as
\begin{quote}
The simplest quantum ontology is {\color{gray}that of the Everett
or Many-Worlds interpretation,} based on a vector in Hilbert space and a Hamiltonian.
Typically one also relies on some classical structure, such as space and local configuration variables within it, which then gets promoted to an algebra of preferred observables. We argue that even such an algebra is unnecessary, and the most basic description of the world is given by the spectrum of the Hamiltonian (a list of energy eigenvalues) and the components of some particular vector in Hilbert space. Everything else --- including space and fields propagating on it --- is emergent from these minimal elements.
\end{quote}

\begin{note}
\label{note:MWI}
I grayed out the mention of Everett's theory in this quote, because there are supporters of \HSF who don't adhere to this interpretation or are agnostic about the status of \qtref{QT:PROJ}.
Carroll \& Singh define HSF as \qtref{QT:HSF}$\Rightarrow$\qtref{QT:TPS}$+$\qtref{QT:OBS}$+$\qtref{QT:PROJ}, but here we'll discuss the more general form~\eqref{eq:HSF-QT}.

In the present article, ``complete description of the reality'' doesn't have the same meaning as in \citep{EPR35}, except for those who support both Everett's theory and \HSF. I'm not discussing interpretations of quantum mechanics here, I prefer to keep \qtref{QT:PROJ} independent, for generality and because it requires a whole separate analysis.
\end{note}

Briefly, \HSF is the extremely Zen \citep{Aaronson2021TheZenAntiInterpretationOfQuantumMechanics} thesis that all we need to describe the world are the Hilbert space $\hilbert$, the Hamiltonian $\oper{H}$, and the unit vector $\ket{\psi}$.

All these without pre-assuming the wavefunction $\braket{\x}{\psi(t)}$ or the matrix form of the Hamiltonian $\mel{\x}{\oper{H}}{\x'}$ or a functional form like $\oper{H}\(\widehat{\x},\widehat{\p}\)$. These shalt not be used until preferred observables like $\widehat{\x}$ and $\widehat{\p}$ emerge from structure alone.
And if we don't get these observables in the exact same way as in standard quantum theory, maybe it's even better, we may get instead quantum gravity \citep{CarrollSingh2019MadDogEverettianism}.

Here is the statement of \HSF that I will use in this article:
\setdefCustomtag{HSF}
\begin{defCustom}[Hilbert space fundamentalism]
\label{def:HSF}
Two triples $\(\hilbert,\oper{H},\ket{\psi}\)$ and $\(\hilbert',\oper{H'},\ket{\psi'}\)$ depict one and the same physical reality if and only if they are isomorphic, in the sense that there is a unitary map between Hilbert spaces $\oper{U}:\hilbert\to\hilbert'$ so that
\begin{equation}
\label{eq:HSF-iso}
\begin{cases}
\oper{H'}&=\oper{U}\oper{H}\oper{U}^{-1} \\
\ket{\psi'}&=\oper{U}\ket{\psi}. \\
\end{cases}
\end{equation}
\end{defCustom}

The main necessary steps in a program to eliminate \qtref{QT:TPS} \& \qtref{QT:OBS} as redundant are
\begin{enumerate}
	\item 
Obtaining a preferred TPS. This was thought to be achieved, in finite dimension, in \citep{CotlerEtAl2019LocalityFromSpectrum}, but that result was not in the sense needed for \HSF.
	\item 
Obtaining the canonically conjugate pairs of observables for the resulting subsystems. For their finite dimensional version see \citep{SinghCarroll2018ModelingPositionAndMomentumInFiniteDimensionalHilbertSpacesViaGeneralizedPauliOperators}.
\end{enumerate}

\subsection{Incompleteness of Hilbert space fundamentalism}
\label{s:intro:HSF-no}

After I presented the quite natural arguments that seem to lead unavoidably to \HSF, let me add that we already have all the elements to refute it as able to give a complete description of reality, as a theory of everything, or even as capable to describe physical change.

Consider automorphisms $\oper{U}_t$ of $\(\hilbert,\oper{H},\ket{\psi}\)$, where $\oper{U}_t=e^{-\frac{i}{\hbar}\oper{H}t}$. Then, equation~\eqref{eq:HSF-iso} becomes
\begin{equation}
\label{eq:HSF-auto-time}
\begin{cases}
\oper{H}&=\oper{U}_t\oper{H}\oper{U}_t^\dagger \\
\ket{\psi(t)}&=\oper{U}_t\ket{\psi(0)}. \\
\end{cases}
\end{equation}

We see that the triple $\(\hilbert,\oper{H},\ket{\psi(t)}\)$ is isomorphic with the triple $\(\hilbert,\oper{H},\ket{\psi(0)}\)$ \citep{Stoica2026NoChangeInHilbertSpaceFundamentalism}.
But $\(\hilbert,\oper{H},\ket{\psi(t)}\)$ can be understood both as equivalent with $\(\hilbert,\oper{H},\ket{\psi(0)}\)$, and as its time evolved state.
It follows straight from \HSF that its description of the world at the time $0$ is the same as that at any other time $t$!

This is quite an ambiguity, how can the world as it is now be isomorphic with the world as it was during Genghis Khan or the dinosaurs or even the big bang, and also with the world as it will be any time in the future?

How completely can $\(\hilbert,\oper{H},\ket{\psi}\)$ describe the world, even if supplemented with various constructions of emergent structures determined by this triple? Wouldn't any such construction, if successful, give identical descriptions for the world at time $t=0$ as for the world at any other time? So \HSF can't even describe changes.

And the ambiguity in the description of the world at different times is just an example, the most evident ambiguity of \HSF, based on the fact that the generator of $\oper{U}_t=e^{-\frac{i}{\hbar}\oper{H}t}$ commutes with $\oper{H}$. But there are infinitely many such generators. For an $N$-dimensional Hilbert space, the possible automorphisms form a space whose dimension is between $N$ and $N^2$.

A too quick reply may be ``c'mon, the wavefunction describes everything about the world!''. But we don't have the wavefunction in \HSF, remember? The wavefunction $\psi(\x,t)$ can extracted from $\ket{\psi(t)}$ using $\psi(\x,t)=\braket{\x}{\psi(t)}$, as in equation~\eqref{eq:wavefunction}. But this requires to know which are the position observables, and \HSF threw away this information when it abandoned \qtref{QT:OBS}.

All we have in \HSF is a unit vector, which is like any other unit vector in $\hilbert$, and the Hamiltonian $\oper{H}$, which is the same like any other operator $\oper{U}\oper{H}\oper{U}^\dagger$.

Another quick reply would be ``the physics is the same, we have all we need, the rest of the details can be obtained by applying the observables to the state vector $\ket{\psi(t)}$''. But, again, in \HSF, while it's true that along with $\hilbert$ we get for free all Hermitian and unitary operators on $\hilbert$, we don't know which of them correspond to the physical properties they're supposed to describe.
Remember that we got rid of \qtref{QT:OBS}.

``So what?'', one may insist, ``there can't be many operators with the same properties like the position operators!''. But there are. For each $\widehat{x}_{k,j}$, where $j\in\{1,2,3\}$ index the space coordinates and $k$ the particle, there are infinitely many other operators with the same spectrum, all of the form $\oper{S}\widehat{x}_{k,j}\oper{S}^\dagger$.
They even stand in the same relation with $\oper{H}$ if $\oper{S}$ commute with $\oper{H}$.
Such an operator $\oper{S}$ turns all of them into other operators, while maintaining the exact same relations between them, and leaving unchanged the form of the Hamiltonian in terms of position and momentum operators \citep[\S2.2]{Stoica2022SpaceThePreferredBasisCannotUniquelyEmergeFromTheQuantumStructure}.
Just like in Euclidean geometry an isometry turns any triangle into a congruent triangle.

``OK, gotcha!'' one may reply ``Then they \emph{do} describe the same physics, because your $\ket{\psi}$ turns together with $\ket{\x}$, so $\psi(\x,t)$ remains unchanged, because it is just $\braket{\x}{\psi(t)}$!''.
But this assumes that we already fixed the position operators, while in \HSF this relation isn't there to begin with.
With infinitely many equivalent choices available, there is no way to fix it without a choice that amounts to what \qtref{QT:OBS} does, and this is what \HSF gives up in the first place.

And if we think a unique choice emerges, since the triples $\(\hilbert,\oper{H},\ket{\psi(t)}\)$ and $\(\hilbert,\oper{H},\ket{\psi(0)}\)$ are equivalent, we should get the same emergent structures at all times.
Either we get a unique one, which can't be used to describe changing things, or many of them, and we have to fix one.

So if the triples from equation~\eqref{eq:HSF-iso} describe the same world, equation~\eqref{eq:HSF-auto-time} implies that \HSF's description of the world doesn't change in time.
\HSF can't have it both ways.

That's why the founders of quantum mechanics had to include \qtref{QT:TPS} and \qtref{QT:OBS} along with the triple $\(\hilbert,\oper{H},\ket{\psi}\)$ \citep{vonNeumann1955MathFoundationsQM,Dirac1958ThePrinciplesOfQuantumMechanics}. Without them, there is no complete description of reality, and in fact, no description of change whatsoever.

I think this refutation of \HSF is sufficient.
There are so many completely different ways to complete the triple $\(\hilbert,\oper{H},\ket{\psi}\)$ with the missing data, that you won't only get different worlds following the same laws, but even different laws \citep{Stoica2022SpaceThePreferredBasisCannotUniquelyEmergeFromTheQuantumStructure,Stoica2024EmpiricalAdequacyOfTimeOperatorCC2HamiltonianGeneratingTranslations}.

But the \HSF camp has a different view on this. They see \qtref{QT:TPS} and \qtref{QT:OBS} as crutches for those who didn't understand \HSF yet.
And this is what this paper is about. To me, it looks like \HSF wants to have it both ways. To me, it all boils down to wanting all isomorphisms like in equation~\eqref{eq:HSF-iso} to describe the same reality, and at the same time automorphisms like in~\eqref{eq:HSF-auto-time}, which is a particular case of~\eqref{eq:HSF-iso}, to describe different realities.

\subsection{So what's this comment all about?}
\label{s:intro:this}

In \citep{Stoica2022SpaceThePreferredBasisCannotUniquelyEmergeFromTheQuantumStructure} I showed that the \HSF program can't recover uniquely the promised structures, in particular the TPS, the 3D space, a preferred pointer basis, or classicality. More precisely, if such a structure is unique, it would not work as needed in quantum theory.
To make the proof fully general, I developed a framework and applied it to several structures that are hoped to emerge in some research programs. What resulted was quite abstract and long, so I added more concrete and simpler results and examples \citep{Stoica2024DoesTheHamiltonianDetermineTheTPSAndThe3dSpace,Stoica2023PrinceAndPauperQuantumParadoxHilbertSpaceFundamentalism}.
Even including a preferred TPS along with $\oper{H}$ and $\ket{\psi}$ is insufficient to recover a complete description of the world \citep{Stoica2023PrinceAndPauperQuantumParadoxHilbertSpaceFundamentalism,Stoica2023AreObserversReducibleToStructures,Stoica2025WhatMakesYouAnObserver}.

Recently, \citep{SoulasFranzmannDiBiagio2025OnTheEmergenceOfPreferredStructuresInQuantumTheoryV2} carefully studied these results, claimed that they found loopholes in the proof, and believed to exploit these alleged loopholes to produce a counterexample.
The main purpose of this comment is to show exactly where their TPS construction fails.
Second, to show that the results from \citep{Stoica2022SpaceThePreferredBasisCannotUniquelyEmergeFromTheQuantumStructure} don't have the pitfalls Soulas et al. claimed to have identified; they were present only in the review of these results as given in \citep{SoulasFranzmannDiBiagio2025OnTheEmergenceOfPreferredStructuresInQuantumTheoryV2}, but not in the original papers.
Third, that the claimed corrections and generalizations don't avoid the problem, and they were already accounted for, and that the counterexample TPS constructed in \citep{SoulasFranzmannDiBiagio2025OnTheEmergenceOfPreferredStructuresInQuantumTheoryV2} doesn't work like the TPS from quantum theory, but is a good toy to illustrate where \HSF fails.
At the same time, equally important, to explain all of these in a friendly but complete manner, even if it will be more slow-paced and longer than an usual reply. This analysis uses Soulas et al.'s TPS as an illustration, but the same method applies to all candidate emergent structures in \HSF that yield observable physical differences.

I also want to clarify the correct meaning of the relational perspective, and to provide an escape for what can be saved from Hilbert space fundamentalism by suggesting a weaker version. I'll use as an illustration the result from \citep{CotlerEtAl2019LocalityFromSpectrum}.

\section{Basic obstructions for the emergence of a unique TPS}
\label{s:TPS}

\subsection{What's a tensor product structure and what's its role?}
\label{s:TPS-basics}

The decomposition \eqref{eq:hilbert-tensor} is referred to as the \emph{tensor product structure} (TPS) of $\hilbert$.
If the system consists of subsystems, for example particles whose number remains constant, the Hilbert spaces $\hilbert_k$ correspond to these subsystems. Decompositions like \eqref{eq:hilbert-tensor} are also proposed to correspond to regions of space \citep{CarrollSingh2019MadDogEverettianism}.

The TPS specifies subsystems, it determines the interactions between subsystems encoded in $\oper{H}$, and the entanglement of the subsystems of $\ket{\psi}$.

Here and in the rest of the paper we will consider the finite-dimensional case, which is also the focus of \citep{SoulasFranzmannDiBiagio2025OnTheEmergenceOfPreferredStructuresInQuantumTheoryV2}.

Let $\hilbert_1,\hilbert_2,\ldots,\hilbert_n$ and $\hilbert$ be complex Hilbert spaces of respective dimensions $d_1,d_2,\ldots,d_n,$ and $N=d_1 d_2\ldots d_n$. A Hilbert space \emph{isomorphism} (unitary operator between two Hilbert spaces)
\begin{equation}
\label{eq:unitary-to-TPS}
\obs{T}:\hilbert_1\otimes\hilbert_2\otimes\ldots\otimes\hilbert_n\to\hilbert
\end{equation}
determines a \emph{tensor product structure} (TPS) on $\hilbert$.

Note that the TPS shouldn't change if we apply a unitary transformation on each individual factor $\hilbert_j$, so the same TPS is defined by $\obs{T}\(\oper{U}_1\otimes\ldots\otimes \oper{U}_n\)$, where $\oper{U}_k$ is a unitary operator on $\hilbert_k$ for all $k$. We call the operators $\oper{U}_k$ \emph{local unitary operators}. The TPS shouldn't change even if we permute factors having the same dimension. So we can also compose $\obs{T}$ with a permutation $\varsigma$ of $(1,\ldots,n)$ satisfying $\dim\hilbert_k=\dim\hilbert_{\varsigma(k)}$, \ie we compose it with unitary operators $\oper{U}_k:\hilbert_k\to\hilbert_{\varsigma(k)}$.
So while the operators $\obs{T}$ and $\obs{T}\(\oper{U}_1\otimes\ldots\otimes \oper{U}_n\)$ may be very different, they define the same TPS. This is why, whenever $\obs{T}'=\obs{T}\(\oper{U}_1\otimes\ldots\otimes \oper{U}_n\)$ like this, we say that $\obs{T}'$ and $\obs{T}$ are equivalent, and denote this by $\obs{T}'\sim\obs{T}$.

\begin{definition}
\label{def:TPS}
A \emph{tensor product structure} is an equivalence class $\mc{T}=[\obs{T}]$ of Hilbert space isomorphisms as in equation~\eqref{eq:unitary-to-TPS}, under the equivalence relation defined as $\obs{T}\sim\obs{T}'$ if and only if $\obs{T}^{-1}\obs{T}'=\oper{U}_1\otimes\ldots\otimes \oper{U}_n$ for some permutation $\varsigma$ of $(1,\ldots,n)$ and unitary operators $\oper{U}_k:\hilbert_k\to\hilbert_{\varsigma(k)}$.
\end{definition}

All unitary transformations of the form $\obs{T}\oper{U}_1\otimes\ldots\otimes \oper{U}_n\obs{T}^{-1}$, including those permuting the factor spaces, form a subgroup of the unitary group $\U\(\hilbert\)$ which leaves $\mc{T}$ unchanged. We denote this subgroup by $\Stab(\mc{T})$ (because it is the \emph{stabilizer} of the TPS $\mc{T}$).

\begin{definition}[Entanglement]
\label{def:entangled-wrt-TPS}
Let $\mc{T}$ be a TPS on $\hilbert$, and $\obs{T}\in\mc{T}$ as in equation~\eqref{eq:unitary-to-TPS}. A state vector $\ket{\psi}$ is a \emph{product state} (or \emph{separable state}) with respect to $\mc{T}$ if and only if $\obs{T}^{-1}\ket{\psi}=\ket{\psi_1}\otimes\ldots\otimes\ket{\psi_n}$, where all $\ket{\psi_k}\in\hilbert_k$. Otherwise, we say that the state $\ket{\psi}$ is \emph{entangled} with respect to $\mc{T}$.
\end{definition}

The notions of product states and entangled states with respect to a TPS from Definition~\ref{def:entangled-wrt-TPS} are well defined, in the sense that they don't depend of the choice of $\obs{T}\in\mc{T}$.
But they depend dramatically of the choice of the TPS $\mc{T}$, and this difference is physically observable.

\subsection{How many physically distinct TPSs are there?}
\label{s:TPS-physical-how-many}

It is evident that the only God-given TPS of a Hilbert space with no other structure is the trivial one, obtained for $n=1$. But how many TPSs are there for fixed dimensions $d_1,d_2,\ldots,d_n$ and $\dim\hilbert=N=d_1 d_2\ldots d_n$?
To find out, we need to ``count'' them up to the contributions of unitary operators $\oper{U}_1,\ldots,\oper{U}_k$ from the unitary transformations $\obs{T}\in\U(\hilbert)$, since their presence leaves the TPS unchanged. The number of possible permutations is finite, so we can ignore it, since the other parameters are continuous. But we also need to take into account that local phase changes are equivalent to each other and to global ones: for any complex number $c$ with $\abs{c}=1$,
$c\obs{T}\(\oper{U}_1\otimes\ldots\otimes\oper{U}_n\)=\obs{T}\(c\(\oper{U}_1\otimes\ldots\otimes\oper{U}_n\)\)= \obs{T}\(\(c\oper{U}_1\)\otimes\ldots\otimes\oper{U}_n\)=\ldots=\obs{T}\(\oper{U}_1\otimes\ldots\otimes\(c\oper{U}_n\)\)$.
Putting these together, we get that we need to find out the dimension of the manifold $\SU(\hilbert)/{\(\SU(\hilbert_1)\times\ldots\times\SU(\hilbert_n)\)}$, where $\SU(\hilbert)$ is the \emph{special unitary group} of $\hilbert$, the group of unitary transformations with determinant $1$. Since $\dim\SU(N)=N^2-1$, the possible distinct tensor product structures form a manifold of dimension $d_1^2 d_2^2\ldots d_n^2-d_1^2- d_2^2-\ldots -d_n^2+n-1$. For example, if all $d_k=d$, we get that the dimension of the space of distinct TPSs is $d^{2n}-nd^2+n-1$. Therefore, the number of parameters that specify a TPS grows exponentially with the number $n$ of qudits.

There is no way to choose a special TPS and call it emergent just from $\hilbert$. At least not without using some information from $\oper{H}$ and maybe $\ket{\psi}$.

But maybe two distinct TPSs $\mc{T}\neq \mc{T}'$ can be physically equivalent. Maybe we can weaken the equivalence relation $\sim$, so that we get a unique TPS. The following result shows that if we do this, physics has a word to say, and that word is ``no!''.
\begin{proposition}
\label{thm:TPS-physical}
For any distinct TPSs $\mc{T}\neq \mc{T}'$ on the Hilbert space $\hilbert$, there is a state vector $\ket{\psi}\in\hilbert$ which is a product state with respect to $\mc{T}$ and an entangled state with respect to $\mc{T}'$.
\end{proposition}
\begin{proof}
We will prove this by contradiction. Suppose that for all $\ket{\psi}\in\hilbert$, $\ket{\psi}$ is a product state w.r.t. $\mc{T}$ if and only if it is a product state w.r.t. $\mc{T}'$.
Let $\obs{T}\in\mc{T}$ and $\obs{T}'\in\mc{T}'$. Then, $\obs{T}^{-1}\obs{T}'$ maps product states from $\hilbert_1\otimes\ldots\otimes\hilbert_n$ into product states from $\hilbert_1\otimes\ldots\otimes\hilbert_n$.
Then, according to a theorem by \citet{Westwick1967TransformationsOnTensorSpaces} (or Lemma A.1 from \citealp{SoulasFranzmannDiBiagio2025OnTheEmergenceOfPreferredStructuresInQuantumTheoryV2}), $\obs{T}^{-1}\obs{T}'=\oper{U}_1\otimes\ldots\otimes \oper{U}_n$ for some permutation $\varsigma$ of $(1,\ldots,n)$ and some operators $\oper{U}_k:\hilbert_k\to\hilbert_{\varsigma(k)}$. Since the product of these operators can only be unitary, one obtains $\mc{T}=\mc{T}'$.
This proves the result by contradiction.
\end{proof}

So we can't get rid of the underdetermination of the TPS by considering different TPSs to be equivalent, because they are not physically equivalent.

Still too many TPSs. So we will try the next best thing, which is to construct the TPS from the other God-given ingredients, the Hamiltonian $\oper{H}$ and, if necessary, also the state vector $\ket{\psi}$.

\subsection{TPS from the Hamiltonian?}
\label{s:TPS-H}

We'll try first to obtain a preferred TPS by using as an ingredient only the Hamiltonian $\oper{H}$.
Suppose we have a procedure, algorithm, method to construct such a TPS.
One should be able to encode such a procedure in a look-up table, a function that associates to the Hamiltonian $\oper{H}$ and the list of dimensions $d_1,d_2,\ldots,d_n$ an isomorphism $\obs{T}$ as in equation~\eqref{eq:unitary-to-TPS},
\begin{equation}
\label{eq:H2tps}
\tau(\oper{H},d_1,d_2,\ldots,d_n)=\obs{T}.
\end{equation}

Then, $\obs{T}$ determines the TPS by the equivalence from Definition~\ref{def:TPS}.
Our procedure can depend of $\oper{H}$, but not of a particular basis of $\hilbert$.
That is, for any unitary $\oper{S}\in\U(\hilbert)$, $\tau$ has to be consistent with transforming both $\oper{H}$ and the TPS in sync using $\oper{S}$,
\begin{equation}
\label{eq:H2tps-invar}
\tau(\oper{S}\oper{H}\oper{S}^\dagger,d_1,d_2,\ldots,d_n)=\oper{S}\obs{T}\in[\obs{T}].
\end{equation}

In particular, any unitary transformation $\oper{S}$ of $\hilbert$ that leaves $\oper{H}$ unchanged, $[\oper{S},\oper{H}]=0$, has to give the same TPS,
\begin{equation}
\label{eq:H2tps-invar-H}
\oper{S}\obs{T}\in[\obs{T}].
\end{equation}

That is, $\oper{S}\obs{T}=\obs{T}\(\oper{U}_1\otimes\ldots\otimes\oper{U}_n\)$, with $\oper{U}_k$ as in Definition~\ref{def:TPS}, so $\oper{S}=\obs{T}\(\oper{U}_1\otimes\ldots\otimes\oper{U}_n\)\obs{T}^{-1}$.
This implies that it has the form $\oper{S}=\wt{\oper{U}}_1\otimes_{\mc{T}}\ldots\otimes_{\mc{T}}\wt{\oper{U}}_n$, where the unitary operators $\wt{\oper{U}}_k$ are local with respect to the TPS $\mc{T}$, possibly permuting its factors.

Then, to ensure that the TPS is unique, the symmetries of the Hamiltonian have to be among the symmetries of the TPS. But simple dimension counting shows that this is impossible. The dimension of the stabilizer group of $\oper{H}$, $\Stab(\oper{H}):=\left\{\oper{U}\in\U(\hilbert)\middle|\oper{S}\oper{H}\oper{S}^\dagger=\oper{H}\right\}$, is between $N=d_1 d_2 \ldots d_n$ and $N^2=d_1^2 d_2^2\ldots d_n^2$, where the lowest bound is achieved when all eigenvalues of $\oper{H}$ are distinct, and the upper bound when they are equal.
But the dimension of the stabilizer of the TPS, $\Stab(\mc{T})$, is, as seen above, $d_1^2+d_2^2+\ldots+d_n^2-n+1$. For example, if all $d_k=d$, we get that the dimension of $\Stab(\oper{H})$ is between $d^n$ and $d^{2n}$, so it grows exponentially. But the dimension of $\Stab(\mc{T})$ is $nd^2-n+1$. Therefore, 

\begin{theorem}
\label{thm:no-unique-TPS-from-H}
There is no procedure to obtain a unique TPS from the Hamiltonian alone.
\end{theorem}

Whenever a TPS is obtained, infinitely many other distinct TPSs can be obtained with the same procedure (unless, of course, $n=1$). The possible distinct TPSs form a space whose dimension grows exponentially with the number $n$ of qudits.

This fact, that the number of parameters needed to specify the TPS grows exponentially, was proved in \citep{Stoica2024DoesTheHamiltonianDetermineTheTPSAndThe3dSpace}. \citep{SoulasFranzmannDiBiagio2025OnTheEmergenceOfPreferredStructuresInQuantumTheoryV2} acknowledged this and the proof can also be found in \emph{``2.2.2 $\widehat{H}$ has more symmetries than $\mc{T}$''} in their paper.

\begin{remark}
\label{rem:cotleretal-TPS}
How can this be possible, when \citep{CotlerEtAl2019LocalityFromSpectrum} gave a construction of a unique TPS from the Hamiltonian's spectrum? In fact, their result is not about the uniqueness of the TPS, it's about the uniqueness of the form of the Hamiltonian when expressed in terms of local operators under the constraint that it couples a small number of them.

Proposition~\ref{thm:TPS-physical} shows that physics disagrees with using equivalence classes of TPSs as TPSs, including the TPS from \citep{CotlerEtAl2019LocalityFromSpectrum}.
To choose one of their TPSs, one also needs to fix a number of parameters exponential in the number of factors $n$.

So their result is different, but it was misunderstood as being about the uniqueness of the TPS itself. The above non-uniqueness proof covers the non-uniqueness of their TPS too.
But despite this, Cotler et al.'s construction has other virtues, which I will discuss in subsections~\sref{s:cotler-emergence-example} and~\sref{s:world-v-law-cotleretal}.
\end{remark}

What else can we try to get rid of this underdetermination of the TPS?
What if another God given structure can partner in with the missing parameters? 
The only other structure in the Hilbert space fundamentalism is the state vector. It has exactly $N$ components, more than enough if the Hamiltonian's eigenvalues are all distinct.
This seems like a match made in heaven, but we still need a way to map the pair $(\oper{H},\ket{\psi})$ to a unique TPS, similar to the map from equation~\eqref{eq:H2tps}:
\begin{equation}
\label{eq:Hv2tps}
\tau(\oper{H},{\color{red}\ket{\psi}},d_1,d_2,\ldots,d_n)=\mc{T}.
\end{equation}

The TPS \citep{SoulasFranzmannDiBiagio2025OnTheEmergenceOfPreferredStructuresInQuantumTheoryV2} can be understood as constructing such a map.

\section{The TPS construction done by Soulas et al.}
\label{s:TPS-soulas}

Now we move to considering the state vector $\ket{\psi}$ as an input to restrict the large number of different TPSs determined by the Hamiltonian $\oper{H}$ only.
The goal is to find a unique map $\tau(\oper{H},\ket{\psi},d_1,d_2,\ldots,d_n)=\mc{T}$.
This is achieved by Soulas et al.'s construction.

\subsection{Soulas et al.'s construction of a unique TPS}
\label{s:TPS-soulas-construction}

Soulas et al. considered a Hamiltonian with distinct eigenvalues $\omega_1<\ldots<\omega_N$, and an eigenbasis $(\ket{\omega_1},\ldots,\ket{\omega_N})$.
This basis and the state vector $\ket{\psi}$ determine $N$ invariants, $\abs{\braket{\omega_1}{\psi}}^2$, $\ldots$, $\abs{\braket{\omega_N}{\psi}}^2$.
Each basis vector $\ket{\omega_j}$ is unique up to some phase factor $e^{i\alpha_j}$, but this phase factor doesn't contribute to the invariant $\abs{\braket{\omega_j}{\psi}}^2$.
Maybe these invariants can be used to fix the remaining $N$ parameters.

Indeed they can, assuming that all $\braket{\omega_j}{\psi}\neq 0$, and  with the help of additional constants.

Let's state explicitly the conditions of Theorem 3.9 from \citep{SoulasFranzmannDiBiagio2025OnTheEmergenceOfPreferredStructuresInQuantumTheoryV2}:
\begin{equation}
\label{eq:soulas-conditions}
\begin{cases}
\omega_1<\ldots<\omega_N \\
\braket{\omega_j}{\psi}\neq 0 \text{ for all }j. \\
\end{cases}
\end{equation}

Soulas et al. construct their TPS in the proof of Theorem 3.9, based on two key insights.

\begin{insight}
\label{insight:entropies}
We need a way to identify a unique TPS.
Recall that the TPS $\mc{T}$ determines the subsystems of $\ket{\psi}$, but also the von Neumann entropy of each subsystem $k$,
\begin{equation}
\label{eq:entropy-TPS}
S(\rho^{\mc{T}}_k)=-\tr\(\rho^{\mc{T}}_k\ln \rho^{\mc{T}}_k\)\in[0,\ln d_k],
\end{equation}
where $\rho^{\mc{T}}_k=\tr^{\mc{T}}_k\dyad{\psi}$ is the reduced density operator for the $k$-th subsystem in the TPS $\mc{T}$.

The converse is also true: if we know the von Neumann entropies for all unit vectors from $\hilbert$, we can reconstruct the TPS.
So all we need is a list of von Neumann entropies $s_{\ket{\psi'},k}$, for each subsystem $k$ and each unit vector $\ket{\psi'}\in\hilbert$.
\end{insight}

\begin{proof}[Proof of Insight~\ref{insight:entropies}]
Let's show that, if we can reconstruct a TPS from the von Neumann entropies of all unit vectors from $\hilbert$, the result is unique.
It is enough to consider those vectors having von Neumann entropy zero for each subsystem, \ie $s_{\ket{\psi'},k}=0$ for all $k$.
Such vectors are products in the TPS we are looking for.
Then, as in the proof of Proposition~\ref{thm:TPS-physical}, the uniqueness of the TPS follows from a theorem by \citet{Westwick1967TransformationsOnTensorSpaces}, or from \citep[Lemma A.1]{SoulasFranzmannDiBiagio2025OnTheEmergenceOfPreferredStructuresInQuantumTheoryV2}.

The existence of such a TPS requires that all constants $s_{\ket{\psi'},k}$ together are compatible with a TPS. This happens only for specific cases, almost all choices giving incompatible constraints.
Almost all functions $s_{\ket{\psi'},k}$ specify inconsistently the von Neumann entropies, not leading to a solution at all, as acknowledged in \citep{SoulasFranzmannDiBiagio2025OnTheEmergenceOfPreferredStructuresInQuantumTheoryV2}. This is because the TPS is determined by a finite number of parameters, but the function $s$ provides a continuous infinity of parameters $s_{\ket{\psi'},k}$, so most of them must be overdetermined. There is no known general way to specify them consistently.
But Soulas et al. added the condition that the choice of $s_{\ket{\psi'},k}$ is consistently made, and it is correct that there is always such a list of $s_{\ket{\psi'},k}$, for any target TPS.
\end{proof}

So, to reconstruct the TPS, all we need is the list of the von Neumann entropies of the $n$ subsystems of all unit vectors in $\hilbert$.
But all unit vectors are the same, so we need a way to label them distinctly.
A basis gives a unique identifier for each vector, but we are not allowed to rely on an arbitrary basis, because then the construction will not be invariant.
So how to endow the vectors with unique identifiers, without using an arbitrary basis? Soulas et al. found a way:

\begin{insight}
\label{insight:unique-id}
We already have a vector $\ket{\psi}\in\hilbert$.
All we need is a way to specify the ``address'' where any other vector $\ket{\psi'}\in\hilbert$ is located in $\hilbert$, using $\ket{\psi}$ and $\oper{H}$ as reference.
There is a way: we can use a complex polynomial $R\in\C[X]$ to construct another vector $\ket{\psi'}\in\hilbert$ by:
\begin{equation}
\label{eq:polynomial-id}
\ket{\psi'}=R\(\oper{H}\)\ket{\psi}.
\end{equation}

We can then label each vector $\ket{\psi'}$ by a polynomial $R$.
If $\oper{H}$ and $\ket{\psi}$ satisfy the conditions~\eqref{eq:soulas-conditions}, we can label like this any vector from $\hilbert$.
\end{insight}

\begin{proof}[Proof of Insight~\ref{insight:unique-id}]
A polynomial $R\in\C[X]$ has the form $R(X)=a_0+a_1 X+a_2 X^2+\ldots+a_m X^m$ for some $m$, with all $a_j$ complex numbers. Then,
\begin{equation}
\label{eq:polynomial-H}
R(\oper{H})=a_0+a_1 \oper{H}+a_2 \oper{H}^2+\ldots+a_m \oper{H}^m.
\end{equation}

Using $\oper{H}=\sum_j\omega_j\dyad{\omega_j}$, from equation~\eqref{eq:polynomial-H} we get
\begin{equation}
\label{eq:polynomial-H-eigenvalues}
R(\oper{H})=\sum_j R(\omega_j)\dyad{\omega_j}.
\end{equation}

Any complex linear operator $\oper{A}$ that commutes with $\oper{H}$ has the form $\oper{A}=\sum_j a_j\dyad{\omega_j}$, where $a_j\in\C$. Then, we can find a polynomial $R\in\C[X]$ so that $\oper{A}=R(\oper{H})$, by solving the system of $N$ equations $R(\omega_j)=a_j$, $j\in\{1,\ldots,N\}$ for the coefficients of $R$.

Then, since all $\braket{\omega_j}{\psi}\neq 0$, any other vector $\ket{\psi'}\in\hilbert$ can be reached by choosing
\begin{equation}
\label{eq:oper-A-H}
\oper{A}=\sum_j \frac{\braket{\omega_j}{\psi'}}{\braket{\omega_j}{\psi}} \dyad{\omega_j}.
\end{equation}

So all we have to do for this is to solve the system of $N$ equations
\begin{equation}
\label{eq:Rvec-coeff}
\left\{R(\omega_j)=\frac{\braket{\omega_j}{\psi'}}{\braket{\omega_j}{\psi}}\right.,
\end{equation}
to find out the coefficients of $R\in\C[X]$.

This proves the surjectivity of the function
\begin{equation}
\label{eq:Rvec}
\begin{aligned}
v:\C[X] & \to \hilbert \\
v(R) & = R(\oper{H})\ket{\psi}. \\
\end{aligned}
\end{equation}

This allows us to label any vector in $\hilbert$ in terms of $\oper{H}$, $\ket{\psi}$, and $R\in\C[X]$.
\end{proof}

\begin{proof}[Proof of uniqueness of the TPS]
It's time to combine Insights~\ref{insight:entropies} and~\ref{insight:unique-id}.
Soulas et al. supplemented the God-given inputs $\oper{H}$ and $\ket{\psi}$ with a function
\begin{equation}
\label{eq:sRi}
\begin{aligned}
s:\C[X]\times\{1,2,\ldots,n\} &\to \R \\
\(R,k\) & \mapsto s_{R,k}\in[0,\ln d_k]. \\
\end{aligned}
\end{equation}

The polynomials $R\in\C[X]$ are used to identify new vectors $\ket{\psi'}=R\(\oper{H}\)\ket{\psi}$, and $s_{R,k}$ specifies the von Neumann entropy of the $k$-th subsystem of $\ket{\psi'}$.
By prescribing that the TPS is such that the von Neumann entropy of the $k$-th subsystem of $R[\oper{H}]\ket{\psi}$ with respect to the TPS is equal to $s_{R,k}$, if a TPS satisfying all these constraints exists, it is unique.
\end{proof}

In Soulas et al.'s construction of a unique TPS, the map~\eqref{eq:Hv2tps} becomes
\begin{equation}
\label{eq:Hv2tps-soulas}
\tau(\oper{H},\ket{\psi},d_1,d_2,\ldots,d_n,{\color{red}s_{R,k}})=\mc{T}.
\end{equation}

Since the space of polynomials $\C[X]$ is a continuum, the function $s:\C[X]\times\{1,\ldots,n\}\to\R$ is specified by an uncountable infinity of continuous parameters from $[0,\ln d_k]$. \citep{SoulasFranzmannDiBiagio2025OnTheEmergenceOfPreferredStructuresInQuantumTheoryV2} explained that this huge number of conditions can be reduced to a countable infinity by using polynomials with complex rational coefficients and taking $s_{R,k}$ to be continuous.

\begin{question}
\label{q:is-uniqueness-enough}
Now that a unique TPS was constructed, should we consider the case closed?

Is the no-go result from \citep{Stoica2022SpaceThePreferredBasisCannotUniquelyEmergeFromTheQuantumStructure} refuted by a counterexample?
Is \HSF saved?
\end{question}

Note that the result from \citep{Stoica2022SpaceThePreferredBasisCannotUniquelyEmergeFromTheQuantumStructure} doesn't say that there is no procedure to construct a unique TPS (see \ref{remark:unique-exist}).
It says that if you have an invariant way to construct a unique TPS, even up to an equivalence relation, then this TPS is not ``physically relevant'', in a sense defined in \citep{Stoica2022SpaceThePreferredBasisCannotUniquelyEmergeFromTheQuantumStructure}. 
To see this, we'll have to ask ourselves the following question:
\begin{question}
\label{q:physical-relevance}
How does \citeauthor{SoulasFranzmannDiBiagio2025OnTheEmergenceOfPreferredStructuresInQuantumTheoryV2}'s TPS behave in a world that changes in time?
\end{question}

\citeauthor{SoulasFranzmannDiBiagio2025OnTheEmergenceOfPreferredStructuresInQuantumTheoryV2} don't analyze their TPS more, jumping straight to the ``Conclusion'' section.

So let's do this ourselves: since Soulas et al. gave us a concrete toy, let's play with it and see what happens. We'll do this in section~\sref{s:TPS-soulas-nogo}.
But first, we'll take a closer look at Insight~\ref{insight:unique-id} (if you're in a hurry you can skip to section~\sref{s:TPS-soulas-nogo}). Section~\sref{s:soulas-emergence-inconsequential} discusses Insight~\ref{insight:entropies}.

\subsection{Why does Soulas et al.'s construction give a unique TPS?}
\label{s:TPS-soulas-basis}

This subsection can be skipped, but there is something interesting behind Insight~\ref{insight:unique-id}.
If $\oper{H}$ and $\ket{\psi}$ satisfy the conditions~\eqref{eq:soulas-conditions}, which are used in the construction of a preferred TPS from \citep{SoulasFranzmannDiBiagio2025OnTheEmergenceOfPreferredStructuresInQuantumTheoryV2}, a unique basis can be obtained without imposing additional symmetry breaking.

First I'll show where a preferred basis emerges implicitly in  Soulas et al.'s construction of a TPS. Then I'll show a more explicit way to build such a unique basis, which can be used to simplify their proof.

Recall that the construction method from Soulas et al. is based on labeling every vector from $\hilbert$ by a polynomial $R\in\C[X]$.
The coefficients of the polynomial $R$ can be obtained by solving the system~\eqref{eq:Rvec-coeff}. Note that system~\eqref{eq:Rvec-coeff} contains $N$ equations, so it is sufficient to restrict the solutions to polynomials of degree at most $N-1$, so $R\in\C_{N-1}[X]$. 
Then, the restriction of the function $v$ from equation~\eqref{eq:Rvec} to $\C_{N-1}[X]$ is bijective. 
Also, $\C_{N-1}[X]$ is a vector space of dimension $N$.
Moreover, the map $v$ is linear, because for all $a,a'\in\C$ and $R,R'\in\C_{N-1}[X]$,
\begin{equation}
\label{eq:v-is-linear}
v(aR+a'R')=aR(\oper{H})\ket{\psi}+a'R'(\oper{H})\ket{\psi'}=av(R)+a'v(R').
\end{equation}

Therefore, $v$ maps the monomial basis $\(1,X,X^2,\ldots,X^{N-1}\)$ of $\C_{N-1}[X]$ into a unique basis of $\hilbert$, $\(\ket{\psi},\oper{H}\ket{\psi},\ldots,\oper{H}^j\ket{\psi},\ldots,\oper{H}^{N-1}\ket{\psi}\)$.
This basis is not, in general, orthonormal, but it is sufficient to label all vectors from $\hilbert$.
Other special bases can be obtained from special bases of $\C_{N-1}[X]$ like Lagrange, Newton, Bernstein, Hermite, Chebyshev, \etc.

Note that, with the method presented here, we need only those $s(R,k)$ for $R\in\C_{N-1}[X]$. We also only need those polynomials $R$ for which $R(\oper{H})\ket{\psi}$ is a unit vector. This reduces dramatically the set of constraints, and this set can be probably reduced even more, but we will still need infinitely many constraints to determine a unique TPS.

\begin{remark}
\label{rem:alternative-construction-TPS}
A perhaps more direct way to obtain a preferred basis, and then a preferred TPS, is the following.
A non-degenerate Hermitian operator $\oper{H}$ comes with an orthogonal eigenbasis $(\ket{\omega_1},\ldots,\ket{\omega_N})$, which is unique up to phase factors. This follows from the first of the conditions~\eqref{eq:soulas-conditions}. Then, a general form of the eigenbasis is $(e^{i\alpha_1}\ket{\omega_1},\ldots,e^{i\alpha_N}\ket{\omega_N})$. Using the second condition from~\eqref{eq:soulas-conditions}, we can fix the phases $\alpha_j$ uniquely within the interval $[0,2\pi)$ by using the vector $\ket{\psi}$.
To fix the phases, we require that the components $e^{i\alpha_j}\braket{\omega_j}{\psi}$ are real and positive, $e^{i\alpha_j}\braket{\omega_j}{\psi}\in(0,1)$.
This gives $e^{-i\alpha_j}=\frac{\braket{\omega_j}{\psi}}{\abs{\braket{\omega_j}{\psi}}}$. We obtain that, under the conditions~\eqref{eq:soulas-conditions}, $\oper{H}$ and $\ket{\psi}$ determine a unique basis of $\hilbert$.
This can be used to simplify the construction from Insight~\ref{insight:unique-id}.

Can we also replace the step achieved by Insight~\ref{insight:entropies} with a simpler one?
Definitely: now all we have to do is to identify this preferred basis with the tensor product of some bases of each $\hilbert_k$, obtaining an isomorphism $\obs{T}$ as in equation~\eqref{eq:unitary-to-TPS}.
The bases of each $\hilbert_k$ can be chosen freely due to Definition~\ref{def:TPS}, but there is some freedom in choosing the order of the basis vectors when constructing $\obs{T}$. However, unlike the choices of the values $s_{R,k}$, there is a finite number of such choices, given by the permutations.
But I will not use my own construction of a TPS in this analysis, to avoid the trap of picking a bad example.
\end{remark}

The construction from Soulas et al. goes straight to the construction of the TPS, without an explicit construction of a basis, although their method to label the vectors by polynomials defines an implicit special basis of $\hilbert$.

Once we have a preferred basis, we can construct various preferred structures, not just a preferred TPS, for example a preferred basis of pointer states for the subsystems that decohere or perform measurements.

Such a construction is perfectly invariant under the conditions~\eqref{eq:soulas-conditions}, and \emph{assuming fixed $\oper{H}$ and $\ket{\psi}$}.
And it doesn't contradict the results from \citep{Stoica2022SpaceThePreferredBasisCannotUniquelyEmergeFromTheQuantumStructure}.

\section{Is Soulas et al.'s TPS unique \textit{and} physically relevant?}
\label{s:TPS-soulas-nogo}

We finally reached the point where we can analyze how Soulas et al.'s construction illustrates the obstructions to emergent structures brought by physical evidence.

\subsection{Is Soulas et al.'s TPS locked in relation with the state vector?}
\label{s:TPS-soulas-locking}

Relations are a two-way street: if $\ket{\psi}$ constrains the TPS based on fixing the values $s_{R,k}$, the resulting TPS ``reacts'' by constraining, in its turn, the von Neumann entropies of $\ket{\psi}$.
This raises the question
\begin{question}
\label{q:TPS-soulas-locking}
Is the TPS constructed in \citep{SoulasFranzmannDiBiagio2025OnTheEmergenceOfPreferredStructuresInQuantumTheoryV2} in a locked relation with the input state vector $\ket{\psi}$ determining it?
\end{question}

The construction uses the input state vector $\ket{\psi}$ to fix the relation between the invariants $s_{R,k}$ specifying the von Neumann entropies to the states whose subsystems have those entropies.
This implies that if we start the construction using as input another state vector $\ket{\psi'}$, we may get a different TPS $\mc{T}'$ from the TPS $\mc{T}$ constructed starting from $\ket{\psi}$,
\begin{equation}
\label{eq:Hv2tps-soulas-two-states}
\tau(\oper{H},{\color{red}\ket{\psi}},d_1,d_2,\ldots,d_n,s_{R,k})\neq \tau(\oper{H},{\color{red}\ket{\psi'}},d_1,d_2,\ldots,d_n,s_{R,k}).
\end{equation}

If we start the construction with $\ket{\psi}$, any vector $\ket{\psi'}$ is obtained from $\ket{\psi}$ and $\oper{H}$ with the help of a polynomial with complex coefficients $R\in\C[X]$ by $\ket{\psi'}=R\(\oper{H}\)\ket{\psi}$.
The von Neumann entropy of the $k$-th subsystem of $\ket{\psi}$ in the TPS $\mc{T}$ is $s_{1,k}$, because  $\ket{\psi}=R_0\(\oper{H}\)\ket{\psi}$, where $R_0(X)=1$ is a constant polynomial (its degree is $0$).
For another unit vector $\ket{\psi'}$, there is another polynomial $R\in\C[X]$ so that $\ket{\psi'}=R\(\oper{H}\)\ket{\psi}$.
Then, the von Neumann entropy of the $k$-th subsystem of $\ket{\psi'}$ with respect to the TPS $\mc{T}$ has to be $s_{R,k}$.

But if we start the construction from $\ket{\psi'}$, using the same specifications of the von Neumann entropies, we get another TPS $\mc{T}'=\tau(\oper{H},\ket{\psi'},d_1,d_2,\ldots,d_n,s_{R,k})$.
In this TPS, the von Neumann entropies of the subsystems of the input vector $\ket{\psi'}$ are $s_{1,k}$, the same as those of the subsystems of $\ket{\psi}$ in the TPS $\mc{T}$, but different from those of $\ket{\psi'}$ in $\mc{T}$, which are $s_{R,k}$.

So, while the TPS from \citep{SoulasFranzmannDiBiagio2025OnTheEmergenceOfPreferredStructuresInQuantumTheoryV2} is determined uniquely by $\oper{H}$ and $\ket{\psi}$ and $s_{R,k}$, we get a different TPS if we start from another input vector and use the same values $s_{R,k}$.

In other words, the invariant property $P_L$ introduced in \citep{SoulasFranzmannDiBiagio2025OnTheEmergenceOfPreferredStructuresInQuantumTheoryV2} determines the von Neumann entropies of the subsystems in the world to be $s_{1,k}$, regardless of the state of the world.
Whichever state of the world, it has the same von Neumann entropies $s_{1,k}$, simply because its TPS is different, as in equation~\eqref{eq:Hv2tps-soulas-two-states}.

One may object that if the world were in another state $\ket{\psi'}$, the construction of the TPS shouldn't have started with $\ket{\psi'}$, but with $\ket{\psi}$, always with $\ket{\psi}$, and then for $\ket{\psi'}$ we would get different entropies $s_{R,k}$, where $\ket{\psi'}=R\(\oper{H}\)\ket{\psi}$, and the problem is avoided. But there is only one fundamental state vector in \HSF, not two. In the original \HSF, one assumes as fundamental only the triple $\(\hilbert,\oper{H},\ket{\psi}\)$ where $\ket{\psi}$ represents the world. The same state vector $\ket{\psi}$ should be used as an input for any emergent structure. Having the world in a state $\ket{\psi'}$ but using as fundamental input for the TPS another state vector $\ket{\psi}$ doesn't sound like what \HSF claims.

But let's say, for the sake of the argument, that it complies with a weaker version, a modification of \HSF:

\setdefCustomtag{HSF$(s_{R,k},\ket{\psi_{\textnormal{TPS}}})$}
\begin{defCustom}[\HSF $+$ locked TPS]
\label{def:HSF-psiTPSs}
Everything about the world is encoded in
\begin{equation}
\label{eq:HSF-psiTPSs}
\(\hilbert,\oper{H},\ket{\psi}{\color{red},s_{R,k},\ket{\psi_{\textnormal{TPS}}}}\)
\end{equation}
where $\ket{\psi_{\textnormal{TPS}}}$ is the input vector for the TPS, and $\ket{\psi}$ is the actual state vector of the world.
\end{defCustom}

This goes against \HSF, so we may want to avoid using in the modified \HSF another unit vector $\ket{\psi_{\textnormal{TPS}}}$ besides $\ket{\psi}$.
But the alternative is to make $s_{R,k}$ depend of $\ket{\psi}$. Moreover, $s_{R,k}$ should depend of $\ket{\psi}$ in the exact right way so that, no matter which state vector $\ket{\psi}$ we use as an input, the result is the same. So the alternative to \ref{def:HSF-psiTPSs} is another modification of \HSF, which can be stated in the following way:
\setdefCustomtag{HSF$(s_{R,k}(\ket{\psi}))$}
\begin{defCustom}[\HSF $+$ state-dependent invariants]
\label{def:HSF-s-of-psi}
Everything about the world is encoded in
\begin{equation}
\label{eq:HSF-s-of-psi}
\(\hilbert,\oper{H},\ket{\psi}{\color{red},s_{R,k}(\ket{\psi})}\)
\end{equation}
where the invariants $s_{R,k}$ depend (!) on the state vector of the world $\ket{\psi}$.
\end{defCustom}

Then, this modification of \HSF uses a funny notion of invariants that depend of time.

Let's see where we arrived.

First, is there any difference between~\ref{def:HSF-psiTPSs} or~\ref{def:HSF-s-of-psi} and specifying the TPS directly, as in quantum theory? Or as in the following modification of \HSF that includes the TPS?
\setdefCustomtag{HSF+TPS}
\begin{defCustom}
\label{def:HSF_TPS}
Everything about the world is encoded in
\begin{equation}
\label{eq:HSF_TPS}
\(\hilbert,\oper{H},\ket{\psi}{\color{red},\mc{T}}\)
\end{equation}
where $\mc{T}$ is the TPS.
\end{defCustom}

If~\ref{def:HSF_TPS} is redundant, as the original \HSF claims,~\ref{def:HSF-psiTPSs} is more redundant. We clearly need less information to specify the TPS directly than indirectly by specifying the input vector $\ket{\psi_{\textnormal{TPS}}}$ and the values $s_{R,k}$.
The reason is that the same TPS can be specified using as input vector any unit vector from $\hilbert$ and the adapted von Neumann entropies $s_{R,k}$.

And~\ref{def:HSF-s-of-psi} is more redundant even than~\ref{def:HSF-psiTPSs}, because now the values $s_{R,k}(\ket{\psi})$ have to be chosen so that we get the same TPS regardless of the input vector $\ket{\psi}$.
So the values of $s_{R,k}(\ket{\psi})$ have to compensate for the changes of $\ket{\psi}$.

The following subsection reveals similar problems when time evolution is taken into account.

\subsection{Soulas et al.'s TPS \vs the time evolution of entanglement}
\label{s:TPS-soulas-time}

In quantum theory, the von Neumann entropy of subsystems can change in time. Subsystems that are in a separable state can become entangled.

\begin{question}
\label{q:TPS-soulas-time}
Does the TPS constructed in \citep{SoulasFranzmannDiBiagio2025OnTheEmergenceOfPreferredStructuresInQuantumTheoryV2} allow entanglement to evolve?
\end{question}

If we use the same TPS at all times, entanglement is allowed to change. But the construction from \citep{SoulasFranzmannDiBiagio2025OnTheEmergenceOfPreferredStructuresInQuantumTheoryV2} depends on the input vector $\ket{\psi(t)}$, which depends on time.
That is, if we do the construction at different times, we may get different TPSs $\mc{T}(t_1)\neq\mc{T}(t_2)$, because
\begin{equation}
\label{eq:Hv2tps-soulas-two-times}
\tau(\oper{H},{\color{red}\ket{\psi(t_1)}},d_1,d_2,\ldots,d_n,s_{R,k})\neq \tau(\oper{H},{\color{red}\ket{\psi(t_2)}},d_1,d_2,\ldots,d_n,s_{R,k})
\end{equation}

This should not happen, because in quantum theory in the Schr\"odinger picture, as used here, the TPS doesn't change in time.

One may think that, since the time evolution operator $e^{-\frac{i}{\hbar}\oper{H}(t_2-t_1)}$ evolves the state $\ket{\psi(t_1)}$ into $\ket{\psi(t_2)}=e^{-\frac{i}{\hbar}\oper{H}(t_2-t_1)}\ket{\psi(t_1)}$, preserving both $\oper{H}$ and all $\abs{\braket{\omega_k}{\psi(t)}}^2$ in the process, the resulting TPS must be independent of the time-dependent state vector $\ket{\psi(t)}$.
In other words, that it doesn't matter whether we use as an input vector $\ket{\psi(t_1)}$ or $e^{-\frac{i}{\hbar}\oper{H}(t_2-t_1)}\ket{\psi(t_1)}$.
However, this is not always possible, because in general $\ket{\psi(t_2)}\neq\ket{\psi(t_1)}$, so the resulting TPSs will be different as in equation~\eqref{eq:Hv2tps-soulas-two-times}.
It is possible that, for special situations, the TPS $\mc{T}(t)$ is time-independent, but it's not desirable, because this blocks the evolution of entanglement. Here is why.

\begin{remark}
\label{rem:no-interaction}
Suppose that Soulas et al's construction yields a time-independent TPS, in the sense that the TPS is independent of $\ket{\psi(t)}$ for the same set of invariants $s_{R,k}$. This is possible only if the invariants $s_{R,k}$ select a TPS in which $e^{-\frac{i}{\hbar}\oper{H}t}=e^{-\frac{i}{\hbar}\oper{H}_1t}\cdot\ldots\cdot e^{-\frac{i}{\hbar}\oper{H}_Nt}$, where $e^{-\frac{i}{\hbar}\oper{H}_k t}$ is a local unitary operator on $\hilbert_k$. In other words, $\oper{H}$ can be expressed in that TPS as the Kronecker sum of local Hermitian operators, $\oper{H}=\oper{H}_1\otimes\oper{I}_2\otimes\ldots+\oper{I}_1\otimes\oper{H}_2\otimes\ldots+$. This is possible only if the spectrum of $\oper{H}$ consists of all possible products of  the spectra of some local operators $\oper{H}_k$.
For such a Hamiltonian and TPS, entanglement doesn't change in time and there are even no interactions. So such a TPS would not be suited for real-world quantum theory, which allows both interactions and entanglement changes in time.
For more details see \citep[section~III]{Stoica2024DoesTheHamiltonianDetermineTheTPSAndThe3dSpace} or the version of this proof given in \citep[Proposition~2.13]{SoulasFranzmannDiBiagio2025OnTheEmergenceOfPreferredStructuresInQuantumTheoryV2}.
\end{remark}

So let's focus now on the solutions $\mc{T}(t)$ that can change in time while keeping $s_{R,k}$ fixed, as in equation~\eqref{eq:Hv2tps-soulas-two-times}.
If $s_{R,k}$ are the same whenever we construct the TPS, the von Neumann entropies of $\ket{\psi(t_1)}$ in the TPS $\mc{T}_1$ equal the von Neumann entropies of $\ket{\psi(t_2)}$ in the TPS $\mc{T}_2$,
\begin{equation}
\label{eq:Hv2tps-soulas-two-times-entropies}
S\(\tr^{\mc{T}_1}_k\dyad{\psi(t_1)}\)=S\(\tr^{\mc{T}_2}_k\dyad{\psi(t_2)}\)=s_{R,k}.
\end{equation}

This is contradicted, again, by our real-world observations that entanglement can change.

To avoid this, we have two options.
The first option is to fix an absolute time $t_0$ and construct the TPS as in \citep{SoulasFranzmannDiBiagio2025OnTheEmergenceOfPreferredStructuresInQuantumTheoryV2}, based on the input vector $\ket{\psi(t_0)}$, and then declare the resulting TPS as definitive, absolute.
This can be expressed as another modification of \HSF:

\setdefCustomtag{HSF$(s_{R,k},t_0)$}
\begin{defCustom}[\HSF $+$ $s_{R,k}$ $+$ absolute time]
\label{def:HSF-t-abs}
Everything about the world is encoded in
\begin{equation}
\label{eq:HSF-t-abs}
\(\hilbert,\oper{H},\ket{\psi(t)}{\color{red},s_{R,k},t_0}\)
\end{equation}
where $t_0$ is the time when the TPS emerges, using $\ket{\psi(t_0)}$ as input vector. At any time, the TPS should be the one constructed for $\ket{\psi(t_0)}$ and $s_{R,k}$.
\end{defCustom}

This makes the TPS consistent with our observations that the TPS is fixed and entanglement can change in time. 
But \ref{def:HSF-t-abs} has the disadvantage of relying on an absolute moment of time. Such an assumption was never a part of \HSF, and it can hardly be considered invariant.

Moreover, since \ref{def:HSF-t-abs} is a particular case of \ref{def:HSF-psiTPSs}, where $\ket{\psi_{\textnormal{TPS}}}=\ket{\psi(t_0)}$, it has the same problem as~\ref{def:HSF-psiTPSs}: it is just~\ref{def:HSF_TPS} but with a more complicated way to fix the TPS by hand.

The alternative, the other possibility to make the constructed TPS independent of time, is to extend it to all moments of time. But to obtain the same TPS, we have to allow $s_{R,k}$ to depend of time, to avoid freezing the entanglement as in equation~\eqref{eq:Hv2tps-soulas-two-times-entropies}.
In this case, the invariants $s_{R,k}$ not only should change in time, but they should change in the right way so that at any time we obtain the same TPS. In other words, the invariants $s_{R,k}$ should change exactly as needed to compensate the time evolution of $\ket{\psi(t)}$.
Instead of the constraints specified by equation~\eqref{eq:sRi}, we use time-dependent constraints,
\begin{equation}
\label{eq:sRi-t}
\(R,k,{\color{red}t}\)\mapsto s_{R,k}{\color{red}(t)}\in[0,\ln d_k],
\end{equation}
and constrain the constraints themselves so that, \emph{as the system evolves in time, the ``invariant constraints'' $s_{R,k}$ also change} to generate the same TPS at all times.

This gives us another modification of \HSF:
\setdefCustomtag{HSF$(s_{R,k}(t))$}
\begin{defCustom}[\HSF $+$ time-dependent $s_{R,k}$]
\label{def:HSF-t}
Everything about the world is encoded in
\begin{equation}
\label{eq:HSF-t}
\(\hilbert,\oper{H},\ket{\psi(t)}{\color{red},s_{R,k}(t)}\)
\end{equation}
using $\ket{\psi(t)}$ and $s_{R,k}(t)$ as inputs at each time $t$,
where all $s_{R,k}(t)$ depend on time in the exact way needed to compensate the time evolution of $\ket{\psi(t)}$, so that we obtain a time-independent TPS.
\end{defCustom}

But then, again, what does~\ref{def:HSF-t} do better than simply~\ref{def:HSF_TPS}? Isn't it just a very complicated way to fix the TPS, which is time-dependent but compensated to give the same TPS at all times?
And how is this construction invariant?

Note that I don't mean to suggest that any of these modifications of \HSF were even proposed in \citep{SoulasFranzmannDiBiagio2025OnTheEmergenceOfPreferredStructuresInQuantumTheoryV2}. They seem to believe that their construction is within \HSF and that it gives a physically valid TPS for all situations. But this is because they didn't test their construction against quantum theory.
The moment you test the TPS to see if it satisfies minimal conditions like allowing interactions and allowing changes in the entanglement, the construction doesn't hold, and you may try to adjust it. I just listed the alternatives.

\subsection{What to sacrifice: invariance, uniqueness, or physical relevance?}
\label{s:TPS-soulas-trilemma}

As noted above, \citep{SoulasFranzmannDiBiagio2025OnTheEmergenceOfPreferredStructuresInQuantumTheoryV2} present their construction as part of \HSF, not as a modification like~\ref{def:HSF-psiTPSs} and~\ref{def:HSF-s-of-psi}, or~\ref{def:HSF-t-abs} and~\ref{def:HSF-t}. Let me insist that I don't try to attribute to Soulas et al. these modifications of \HSF as a straw man, to win an argument; what I say is that they should have tested their TPS in these situations, and if they would do it, this is what they would get.
But after doing their construction they moved to the Conclusion section, leaving unanswered (and unasked) questions that, if we attempt to answer, we realize that bare \HSF is insufficient. The modifications of \HSF mentioned above are just the alternatives to which asking these questions lead.

Let us review our previous observations.
The first observation is that, as we try to make the TPS accommodate minimal realistic requirements like allowing interactions and allowing changes in entanglement, we realize that we are forced to choose one of the following:
\begin{enumerate}
	\item 
The construction of the TPS has to depend of a fixed, absolute time $t_0$, as in \ref{def:HSF-t-abs}. More precisely, the construction has to depend of $\ket{\psi(t_0)}$, which depends on $t_0$.
	\item 
The construction of the TPS, \emph{including the invariants $s_{R,k}$}, has to depend of time, as in \ref{def:HSF-t} (again, because $\ket{\psi(t)}$ depends on time).
\end{enumerate}

Can any of these options be considered an invariant construction?
One may call the values $s_{R,k}$ invariants, but are they invariant, if they either depend of an absolute time $t_0$, or if they have to change in time in the exact right way to compensate for the changes of the state vector just enough to get the same TPS at all times?

Time translations are among the symmetries of a quantum system, and therefore an invariant should be also invariant to the time translation symmetries. 
Seen as such, the TPS construction from Soulas et al. can't really be invariant, neither the values $s_{R,k}$ deserve to be called like this.
They are invariants if we consider $(\hilbert,\oper{H},\ket{\psi})$ with a fixed $\ket{\psi}$, but $\ket{\psi}$ is in fact time-dependent, otherwise \HSF doesn't describe a world in which there is time.

\begin{question}
\label{q:soulas-TPS-time-translation}
Can we restore time translation symmetry in Soulas et al.'s construction?
\end{question}

Let's start with \HSF with an absolute time $t_0$, as in \ref{def:HSF-t-abs}.
In this formulation, at any time $t$, we consider as TPS the one obtained at the time $t_0$, that is, using as an input vector $\ket{\psi(t_0)}$.
To restore the time symmetry, we have to allow $t_0$ to take any value.
But then, to get a unique TPS, we need to also allow the invariants $s_{R,k}$ to depend of $t_0$, which results in \ref{def:HSF-t}.
If we keep the invariants $s_{R,k}$ invariant in time as well, then the resulting TPSs are different for different input times $t_0$.
Then, the TPS is not unique, since at any time we have all the TPSs obtained at all other times, but at least now the construction is time-independent.

If, on the other hand, we start with \HSF with all times, as in \ref{def:HSF-t}, the solution is again not time-independent, since the invariants $s_{R,k}$ depend of time. To make the construction time-independent, we have to allow time translations, which amounts to shifting $s_{R,k}(t)$ with any possible time interval $\Delta t\in\R$, so we have to accept all solutions $s_{R,k}(t+\Delta t)$.
But shifting $s_{R,k}(t)\mapsto s_{R,k}(t+\Delta t)$ for the same $\ket{\psi(t)}$ results in infinitely many TPSs, because entanglement entropies change continuously in time. For more details, see \citep{Stoica2024DoesTheHamiltonianDetermineTheTPSAndThe3dSpace}, reprised in \citep[Proposition~2.13]{SoulasFranzmannDiBiagio2025OnTheEmergenceOfPreferredStructuresInQuantumTheoryV2}.
So then again the solution is not unique, but at least this construction is time-independent.

Notice how, regardless whether we started with \ref{def:HSF-t-abs} or \ref{def:HSF-t}, making the construction time-independent leads to the same non-uniqueness of the TPS.

Non-uniqueness is fatal, because if there are more TPSs, there is no objective way to define the von Neumann entropies of the subsystems for the state vector $\ket{\psi(t)}$ at the current time. And infinitely many TPSs are obtained by using $\ket{\psi(t_0)}$ for all possible times $t_0$ (respectively by using $\ket{\psi(t+\Delta t)}$ for all times $t+\Delta t$).

To summarize, if the construction of the TPS is invariant (including time-independent),
\begin{enumerate}
	\item If it results in a unique TPS, that TPS doesn't allow the von Neumann entropies change in time, so separable state can't become entangled (Remark~\ref{rem:no-interaction}). Also there are no interactions.
	Alternatively, we can try to fix a TPS once and for all, which again would violate the invariance of the construction.
	\item If the solution is not unique, our many TPSs don't allow us to say anything about the von Neumann entropies, since the answer depends on the TPS. Also we can't say anything about subsystems and their interactions, since they depend of the TPS too.
\end{enumerate}

This leaves us with a trilemma:

\begin{question}
\label{q:TPS-soulas-trilemma}
Can the construction from \citep{SoulasFranzmannDiBiagio2025OnTheEmergenceOfPreferredStructuresInQuantumTheoryV2} be at the same time invariant, unique, and physically relevant?
\end{question}

And this catch-22 is what \thmStoica from \citep{Stoica2022SpaceThePreferredBasisCannotUniquelyEmergeFromTheQuantumStructure} says happens. Not only for the TPS, but for any structure hoped to emerge in \HSF, if that structure should be capable of yielding observable changes in the system (like the TPS with entanglement, 3D space with the distribution of matter, \etc.).

\section{What's (Stoica, 2022)'s no-go theorem really about?}
\label{s:stoica-all-about}

\citep{SoulasFranzmannDiBiagio2025OnTheEmergenceOfPreferredStructuresInQuantumTheoryV2}'s review of the results from \citep{Stoica2022SpaceThePreferredBasisCannotUniquelyEmergeFromTheQuantumStructure} is careful and extensive, but it may leave a wrong impression about those results.
After I explain the main result from \citep{Stoica2022SpaceThePreferredBasisCannotUniquelyEmergeFromTheQuantumStructure}, I compare the original source with the way it was understood by Soulas et al.
I also explain how Soulas et al.'s proposed generalization is a particular case of these results.

\subsection{The main idea of (Stoica, 2022)'s no-go theorem}
\label{s:stoica-brief}

The structures that we use in quantum theory are in various relations with each other. These relations can be time-independent, or they can change in time.

For example, $\oper{H}$ and $\ket{\psi(t)}$ are in a time-independent relation, because we can't observe changing differences between them.
To see this, let the Hamiltonian be
\begin{equation}
\label{eq:H-eigen-decomp}
\oper{H}=\sum_{\omega}\omega\oper{P}_{\omega},
\end{equation}
where $\omega$ are the eigenvalues of $\oper{H}$ and $\oper{P}_{\omega}$ are the projections on the eigenspaces of $\oper{H}$.
Then,
\begin{equation}
\label{eq:psi-H-relation}
\ket{\psi(t)}=\sum_{\omega}\oper{P}_{\omega}\ket{\psi(t)}.
\end{equation}

From equation~\eqref{eq:time-evol} we get, for each $\omega$,
\begin{equation}
\label{eq:psi-H-relation-omega-invariant}
\oper{P}_{\omega}\ket{\psi(t)}=\oper{P}_{\omega}e^{-\frac{i}{\hbar}\oper{H}t}\ket{\psi(0)}=e^{-\frac{i}{\hbar}\omega t}\oper{P}_{\omega}\ket{\psi(0)}.
\end{equation}

Then, for each $\omega$,
\begin{equation}
\label{eq:psi-H-expval}
\expval{\oper{P}_{\omega}}{\psi(t)}=\expval{\oper{P}_{\omega}}{\psi(0)}.
\end{equation}

So $\ket{\psi(t)}$ doesn't change in time with respect to $\oper{H}$ in an observable way. The lengths of the projections $\oper{P}_{\omega}\ket{\psi(t)}$ are time-independent. The phases $e^{-\frac{i}{\hbar}\omega t}$ change at different rates for different eigenvalues $\omega$, but they are not observable.

On the other hand, $\ket{\psi(t)}$ changes in time with respect to other structures. It changes with respect to space, since $\psi(\x,t)=\braket{\x}{\psi(t)}$ changes. It also changes with respect to the TPS, because the subsystems of $\ket{\psi(t)}$ can interact and can change their entanglement.

So a TPS or an emergent space that are not in a time-dependent relation with $\ket{\psi(t)}$ don't have observable effects.
This is why, in \citep{Stoica2022SpaceThePreferredBasisCannotUniquelyEmergeFromTheQuantumStructure}, I called structures that are in an observable time-dependent relation with $\ket{\psi(t)}$ ``physically relevant''.

\thmStoica from \citep{Stoica2022SpaceThePreferredBasisCannotUniquelyEmergeFromTheQuantumStructure} is about the emergence of physically relevant structures $\mc{S}$, \ie whose relation with $\ket{\psi(t)}$ can be different, in particular it can change in time. 

We've seen that, based on the time-dependence of the relation.
\begin{enumerate}
	\item[R1.] $\oper{H}$ and $\ket{\psi(t)}$ are in a time-independent relation.
	\item[R2.] $\mc{S}$ and $\ket{\psi(t)}$ are in a time-dependent relation.
\end{enumerate}

But how can the relation between $\mc{S}$ and $(\oper{H},\ket{\psi(t)})$ be?
Suppose that $\mc{S}$ emerges from $(\oper{H},\ket{\psi(0)})$, \emph{uniquely up to physically unobservable differences}.
Let's denote this by
\begin{equation}
\label{eq:unique-emergence}
(\oper{H},\ket{\psi(0)})\stackrel{!}{\longrightarrow}\mc{S}.
\end{equation}

Then, if the construction of the emergent structure $\mc{S}$ is invariant, a unitary transformation $\oper{U}$ of $\hilbert$ should preserve this construction, and so, from $(\oper{U}\oper{H}\oper{U}^\dagger,\oper{U}\ket{\psi(0)})$, a structure of the same kind as $\mc{S}$, denoted by $\oper{U}\cdot\mc{S}$, should emerge uniquely
\begin{equation}
\label{eq:emergent-invariant-relation}
(\oper{U}\oper{H}\oper{U}^\dagger,\oper{U}\ket{\psi(0)})\stackrel{!}{\longrightarrow}\oper{U}\cdot\mc{S}.
\end{equation}

Therefore,
\begin{enumerate}
	\item[R3.] If $(\oper{H},\ket{\psi(0)})\stackrel{!}{\longrightarrow}\mc{S}$, $\mc{S}$ and $(\oper{H},\ket{\psi(0)})$ are in a unitary-independent relation.
\end{enumerate}

The main theorem from \citep{Stoica2022SpaceThePreferredBasisCannotUniquelyEmergeFromTheQuantumStructure} is:

\setdefCustomtag{Theorem~2}
\begin{defCustom}[\citealp{Stoica2022SpaceThePreferredBasisCannotUniquelyEmergeFromTheQuantumStructure}]
\label{thm:theoremtwo}
If a $\mc{K}$-structure is physically relevant, then it is not (essentially) unique.
\end{defCustom}

Here ``essentially unique'' means that it emerges uniquely up to physically unobservable differences as in equation~\eqref{eq:emergent-invariant-relation}. A $\mc{K}$-structure is a structure $\mc{S}$ of a particular kind $\mc{K}$, which can be TPS, 3D space, pointer basis \etc., along with some invariant conditions that are used in the construction of $\mc{S}$.
Therefore, the theorem can be formulated as ``a structure $\mc{S}$ in a time-dependent relation with $\ket{\psi(t)}$ can't emerge uniquely from $\oper{H}$ and $\ket{\psi(t)}$.''

\begin{proof}[Proof of \thmStoica, main idea]
In particular, equation~\eqref{eq:emergent-invariant-relation} should be true even if we apply as our symmetry transformation a unitary transformation of the form $\oper{U}_t=e^{-\frac{i}{\hbar}\oper{H}t}$, so
\begin{equation}
\label{eq:emergent-invariant-relation-time}
(\oper{H},e^{-\frac{i}{\hbar}\oper{H}t}\ket{\psi(0)})\stackrel{!}{\longrightarrow}e^{-\frac{i}{\hbar}\oper{H}t}\cdot\mc{S}.
\end{equation}

But since the symmetry transformation $e^{-\frac{i}{\hbar}\oper{H}t}$ also plays the role of time evolution operator, uniqueness (up to physically unobservable differences) implies that $\mc{S}$ and $(\oper{H},\ket{\psi(t)})$ are in a time-independent relation,
\begin{enumerate}
	\item[R3$'$.] If $(\oper{H},\ket{\psi(0)})\stackrel{!}{\longrightarrow}\mc{S}$, $\mc{S}$ and $(\oper{H},\ket{\psi(t)})$ are in a time-independent relation.
\end{enumerate}

Since a structure in a time-independent relation with a pair of structures is also in a time-independent relation with each of these structures, we get
\begin{enumerate}
	\item[$\neg$R2.] If $(\oper{H},\ket{\psi(0)})\stackrel{!}{\longrightarrow}\mc{S}$, $\mc{S}$ and $\ket{\psi(t)}$ are in a time-independent relation.
\end{enumerate}

Therefore, we derived a contradiction between R2 and R3. This proves the theorem.
\end{proof}

Notice that $\oper{H}$ and $\ket{\psi(t)}$ are in a time-independent relation, and yet we get to know the Hamiltonian. How can this be possible if there are no observable changes in the relation between $\oper{H}$ and $\ket{\psi(t)}$? The answer is, of course, that we ``read'' the Hamiltonian from the way the wavefunction $\psi(\x,t)=\braket{\x}{\psi(t)}$ changes in time with respect to space and the TPS, not by accessing directly the Hamiltonian or its relation with the unit vector $\ket{\psi(t)}$.

An obvious question is
``if we give up uniqueness, can a physically relevant structure emerge?''

Suppose we give up uniqueness up to physically unobservable differences, replacing it with a weaker equivalence relation ``$\approx$''.
Then, two emergent structures $\mc{S}\approx\mc{S}'$ will give different descriptions of the world.
For example, two emergent TPSs $\mc{T}\approx\mc{T}'$ will give different von Neumann entropies (see Proposition~\ref{thm:TPS-physical}).
Two 3D spaces $\mc{X}\approx\mc{X}'$ will give different distributions of matter in space. 
So the description provided by a structure unique only up to ``$\approx$'' is ambiguous.
Moreover, since the notion of physical uniqueness that we wanted to relax was contradicted by the necessity that the relation between $\mc{S}$ and $\ket{\psi(t)}$ is time-dependent, the weaker equivalence relation ``$\approx$'' has to ignore these differences. In other words, our emergent structure will be so ambiguous that it wouldn't be able to describe change, this time precisely because it is too ambiguous.

Moreover, if our construction is invariant, using the same transformation~\eqref{eq:emergent-invariant-relation-time} it follows that our construction gives at all times the exact same set of emergent structures.
This set give contradictory descriptions of the world, including describing the world at a given time as it is at any other time.

So both uniqueness up to physically unobservable differences and uniqueness up to a weaker equivalence relation ``$\approx$'' have the same consequence that our emergent structure can't describe changes. 
And this is confirmed by the result from subsection~\sref{s:intro:HSF-no}.

And, in section~\sref{s:TPS-soulas-nogo}, we've seen how Soulas et al.'s construction of a TPS illustrates \thmStoica, rather than contradicting it.

\subsection{Does (Stoica, 2022) claim that uniqueness is impossible?}
\label{s:soulas-claim-no-uniqueness}

If a theorem says ``P and Q can't both be true'', or ``If P is true, Q is false'', can it be refuted with an example that shows that Q is true, but says nothing about P? Obviously not. If Q is true, this is perfectly consistent with the theorem, if also P is false. To refute such a theorem, one should give an example that both P and Q are true.

In \thmStoica from \citep{Stoica2022SpaceThePreferredBasisCannotUniquelyEmergeFromTheQuantumStructure}, P=''the emergent structure is physically relevant'' and Q=''the emergent structure is unique''.

However, \citep{SoulasFranzmannDiBiagio2025OnTheEmergenceOfPreferredStructuresInQuantumTheoryV2} think that \citep{Stoica2022SpaceThePreferredBasisCannotUniquelyEmergeFromTheQuantumStructure} says ``Q is false'':

\begin{quote}
we can finally turn back to the initial problem and prove that a pair $(\oper{H},\ket{\psi})$ can uniquely determine a TPS, contradicting one of
Stoica's strongest claims [27] (section 3.3).
\end{quote}

Later, in a talk, Soulas said ``Stoica's strongest claim: no unitary-invariant property on $(\oper{H},\ket{\psi})$ can uniquely determine a structure'' \citep[min. 60]{Soulas2026OnTheEmergenceOfPreferredStructuresInQuantumTheoryTalk}.

In fact, \citep[\ref{remark:unique-exist}]{Stoica2022SpaceThePreferredBasisCannotUniquelyEmergeFromTheQuantumStructure} clearly states that it's perfectly possible to construct unique structures from $\oper{H}$ and $\ket{\psi}$, and then gives some examples:
\setdefCustomtag{Remark~15}
\begin{defCustom}[from \citealp{Stoica2022SpaceThePreferredBasisCannotUniquelyEmergeFromTheQuantumStructure}]
\label{remark:unique-exist}
There are ways to construct structures that depend on the
MQS {\normalfont[\ie $(\hilbert,\oper{H},\ket{\psi})$]} alone and are unique, but they all violate the Physical Relevance \ref{condition-relevance}. Such examples include... {\normalfont[several examples]}
\end{defCustom}

So constructing a unique TPS doesn't, by itself, contradict \citep{Stoica2022SpaceThePreferredBasisCannotUniquelyEmergeFromTheQuantumStructure}, contrary to the claims made by Soulas et al.
To be a counterexample, the TPS should be shown to satisfy both conditions that the theorem says can't both be true. And also the construction has to be invariant, since this is assumed by the theorem.

\subsection{Is (Stoica, 2022) limited to strict, absolute uniqueness?}
\label{s:stoica-not-just-uniqueness}

Soulas et al. also stated that \citep{Stoica2022SpaceThePreferredBasisCannotUniquelyEmergeFromTheQuantumStructure} talks about uniqueness in an absolute sense.
For example, here is what they write in \citep[section~\S3.1]{SoulasFranzmannDiBiagio2025OnTheEmergenceOfPreferredStructuresInQuantumTheoryV2}:
\begin{quote}
In this absolute perspective, there is a continuous infinity of different ways to choose a unit vector in a bare Hilbert space $\hilbert$; whereas in the previous relational perspective, there is only one way to choose it.
\end{quote}

Nowhere in \citep{Stoica2022SpaceThePreferredBasisCannotUniquelyEmergeFromTheQuantumStructure} or elsewhere it is claimed, suggested, or assumed that in a bare Hilbert space there is a physically relevant difference between two unit vectors, just because they are different elements of $\hilbert$. Also the refutation of \HSF is not based on such an assumption, it is based on \emph{finding observable physical differences} between two otherwise isomorphic structures.

Here is another misattribution from \citep[section~\S3.1]{SoulasFranzmannDiBiagio2025OnTheEmergenceOfPreferredStructuresInQuantumTheoryV2}:
\begin{quote}
For Stoica, on the other hand, the structures are perceived in a more absolute sense: $\mc{T}$ and $e^{-i\oper{H}t}\cdot\mc{T}$ are seen as distinct TPSs, regardless of whether they can be distinguished
by some surrounding structure.
\end{quote}

But in \citep{Stoica2022SpaceThePreferredBasisCannotUniquelyEmergeFromTheQuantumStructure}, two TPSs (or other emerging structures) are \emph{not} considered distinct merely because another one can be obtained by a mere unitary transformation; they are distinct \emph{precisely when they can be distinguished physically}.
The symmetry transformation $e^{-i\oper{H}t}\cdot\mc{T}$ is used either to produce physically distinct structures that decode physically different descriptions of reality from $\ket{\psi}$, or to prove that the structure doesn't have different observable effects at different times.
All symmetries that don't result in physical differences are allowed as part of the definition of uniqueness, and no construction of an ``emergent structure'' is rejected just because it has distinct solutions without having distinct physical effects.

The essential keyword in \thmStoica is ``essentially unique'', \citep[page 3915]{Stoica2022SpaceThePreferredBasisCannotUniquelyEmergeFromTheQuantumStructure}:

\setdefCustomtag{Essential uniqueness}
\begin{defCustom}[From \citealp{Stoica2022SpaceThePreferredBasisCannotUniquelyEmergeFromTheQuantumStructure}, {\S}3.2.1.]
\label{quote-essential-uniqueness}
The first condition we will impose on a
$\mc{K}$-structure is to be \emph{unique} or \emph{essentially unique}. We will not require the structure to be necessarily unique, but we require at least that whenever such a structure is not unique, there are no physical differences between its instances.
\end{defCustom}

So in \citep{Stoica2022SpaceThePreferredBasisCannotUniquelyEmergeFromTheQuantumStructure} there is nothing like ``absolute uniqueness'', as claimed by Soulas et al.
I borrowed the term ``essentially unique'' from \citep{CarrollSingh2019MadDogEverettianism}, who used it without defining it to refer to the result from \citep{CotlerEtAl2019LocalityFromSpectrum}.
Since it wasn't clear what this term was supposed to mean, to cover all bases, I allowed ``essentially unique'' to mean ``equivalence up to physically unobservable differences'' (\citealp{Stoica2022SpaceThePreferredBasisCannotUniquelyEmergeFromTheQuantumStructure}, page 3916):

\setdefCustomtag{Condition~1}
\begin{defCustom}[``Essentially unique''-ness from \citealp{Stoica2022SpaceThePreferredBasisCannotUniquelyEmergeFromTheQuantumStructure}]
\label{condition-uniqueness}
Any two $\mc{K}$-structures of the same kind $\mc{K}$ are {\color{red}physically equivalent}.
\end{defCustom}

A nice consequence of this general notion of uniqueness is that the proof doesn't need to restrain physical equivalence to strict uniqueness, precisely because the same equivalence relation used to define ``essential uniqueness'' also appears in the definition of ``physical relevance''. This can be seen in the relation between the notions of time-dependence and physical uniqueness from subsection~\sref{s:stoica-brief}, and in \citep[page 3919]{Stoica2022SpaceThePreferredBasisCannotUniquelyEmergeFromTheQuantumStructure}:

\setdefCustomtag{Condition~2}
\begin{defCustom}[{\color{red}Physical relevance} from \citealp{Stoica2022SpaceThePreferredBasisCannotUniquelyEmergeFromTheQuantumStructure}]
\label{condition-relevance}
There exist at least two unit vectors $\ket{\psi}\nsim\ket{\psi'}\in\hilbert$ representing distinct physical states that are not distinguished by the Hamiltonian, and an invariant scalar function $\mc{I}$ of $\mc{S}_{\widehat{H}}^{\ket{\psi}}$ and $\ket{\psi}$ so that, for any elements $g,g'\in G_P$,
\begin{equation}
\label{eq:physical_relevance_invar_G}
\mc{I}\(\widehat{g}\[\mc{S}_{\widehat{H}}^{\ket{\psi}}\],\ket{\psi}\)
\neq\mc{I}\(\widehat{g'}\[\mc{S}{}_{\widehat{H}}^{\ket{\psi'}}\],\ket{\psi'}\).
\end{equation}
\end{defCustom}

Here, $G_P\subset\U(\hilbert)$ is a symmetry group responsible for physically unobservable differences, for example local gauge transformations. The notation $\mc{S}_{\widehat{H}}^{\ket{\psi}}$ translates into $(\widehat{H},\ket{\psi},\mc{S})$ in \citep{SoulasFranzmannDiBiagio2025OnTheEmergenceOfPreferredStructuresInQuantumTheoryV2}, and all elements transform in sync under a unitary symmetry.

To summarize,
\begin{enumerate}
	\item 
	The results from \citep{Stoica2022SpaceThePreferredBasisCannotUniquelyEmergeFromTheQuantumStructure} are valid even if we consider distinct TPSs as equivalent.
	\item 
	The story doesn't end if one constructs a unique preferred structure, since \thmStoica from \citep{Stoica2022SpaceThePreferredBasisCannotUniquelyEmergeFromTheQuantumStructure} doesn't claim that no unique preferred structure is possible, it claims that \emph{if} such structure is unique, \emph{then} it is not physically relevant.
\end{enumerate}

Then why did \citep{SoulasFranzmannDiBiagio2025OnTheEmergenceOfPreferredStructuresInQuantumTheoryV2} claim that the results from \citep{Stoica2022SpaceThePreferredBasisCannotUniquelyEmergeFromTheQuantumStructure} refer to absolute uniqueness, if this isn't true? In their defense, maybe they assumed that the result from \citep[\S{II}, Theorem 1]{Stoica2024DoesTheHamiltonianDetermineTheTPSAndThe3dSpace} is a restatement of \thmStoica from \citep{Stoica2022SpaceThePreferredBasisCannotUniquelyEmergeFromTheQuantumStructure}. Theorem~1 from the three-page paper \citep{Stoica2024DoesTheHamiltonianDetermineTheTPSAndThe3dSpace}, like Theorem~\ref{thm:no-unique-TPS-from-H} from this article, is about absolute uniqueness, and it relies only on $\oper{H}$, not on $\ket{\psi}$, so in this sense it's a weaker result than \thmStoica. It was intended to show how many different ways to choose the TPS based on only the Hamiltonian there are. And we've seen in Proposition~\ref{thm:TPS-physical} that physics requires strict uniqueness of the TPS.

However, even in that short paper, Section~\S{III} discusses general notions of equivalence:
\begin{quote}
even if we would hope to get rid of non-uniqueness by taking equivalence classes of the resulting TPSs, this would not make sense physically
\end{quote}
followed immediately by a proof that, regardless of the invariant method used to construct the TPS, no form of equivalence of TPSs which is consistent with entanglement entropy can allow entanglement entropy change in time.

\subsection{Intrinsic \vs extrinsic structures}
\label{s:intrinsic-vs-extrinsic}

In \citep{Stoica2022SpaceThePreferredBasisCannotUniquelyEmergeFromTheQuantumStructure} I treat the candidate preferred structures as consisting of tensors over $\hilbert$. This is as general as it gets, if we discuss intrinsic structures, structures constructable within $(\hilbert,\oper{H},\ket{\psi})$. Anything else are structures external to $(\hilbert,\oper{H},\ket{\psi})$, so they can't be a part of \HSF.

A candidate preferred structure has to satisfy some invariant constraints, so I defined the ``kind'' of the structure as a set of invariant conditions over the tensors constituting the structure. The unitary group $\U(\hilbert)$ acts on the structures that satisfy these constraints.
Its action are unitary transformations, not other kinds of transformations outside of $(\hilbert,\oper{H},\ket{\psi})$.

On the other hand, when \citep{SoulasFranzmannDiBiagio2025OnTheEmergenceOfPreferredStructuresInQuantumTheoryV2} reviewed \citep{Stoica2022SpaceThePreferredBasisCannotUniquelyEmergeFromTheQuantumStructure}, they gave a broader definition of the ``kind'', as a set on which $\U(\hilbert)$ acts. Their definition seems more general, in the sense that it can be applied to structures that involve more than the Hilbert space and its elements, and moreover, it allows a free specification of the action of $\U(\hilbert)$.
Their definition allows one to add any structure, even extrinsic ones, and to define freely the action of $\U(\hilbert)$ on that structure.

I think that, if a structure of a particular kind is to emerge from $(\hilbert,\oper{H},\ket{\psi})$, this has to be an intrinsic construction, so in any emergence program we better avoid anything that is not explicitly intrinsic.

But even the actions of $\U(\hilbert)$ on structures external to $(\hilbert,\oper{H},\ket{\psi})$ induce actions on its intrinsic structures, and these are the ones accessible to observations. So I see no problem in leaving unspecified the space of structures on which the group acts, as long as at the end of the construction, the resulting emergent structures are shown to be intrinsic. I see this neither as a problem, nor as a generalization, since ultimately the emergent structure has to be intrinsic.

\subsection{Soulas et al.'s ``relational uniqueness'' is covered by (Stoica, 2022)}
\label{s:soulas-no-generalize}

\citep{SoulasFranzmannDiBiagio2025OnTheEmergenceOfPreferredStructuresInQuantumTheoryV2} defined the product of kinds, to which they add sufficient constraints to determine the structure:

\setdefCustomtag{Soulas et al's ``relational uniqueness''}
\begin{defCustom}[\textbf{Definition 3.3} from \citealp{SoulasFranzmannDiBiagio2025OnTheEmergenceOfPreferredStructuresInQuantumTheoryV2}]
\label{def:soulas:relational-uniqueness}
Let $\mc{K}_0$ and $\mc{K}_e$ be two determined kinds, and $P$ a unitary invariant property on the product kind $\mc{K}_0\times\mc{K}_e$. We say that $P$ determines the product kind if $P$ holds on exactly one orbit in $\mc{K}_0\times\mc{K}_e$ under the action of $\U(\hilbert)$.
\end{defCustom}

Soulas et al. call it ``the correct way to formulate HSF''.
This may seem as different and more general, but it is perfectly contained in the framework from \citep{Stoica2022SpaceThePreferredBasisCannotUniquelyEmergeFromTheQuantumStructure}.

In \citep{Stoica2022SpaceThePreferredBasisCannotUniquelyEmergeFromTheQuantumStructure}, any kind $\mc{K}$ comes with its types of tensors over $\hilbert$, $\ms{T}_{\mc{K}}$, and with a set of invariant constraints $\ms{C}_{\mc{K}}$ that tensors of the specified types should satisfy to constitute a candidate emergent structure (Definition 3, page 3914).

So two kinds $\mc{K}_0$ and $\mc{K}_e$ come with their own types of tensors $\ms{T}_0$ and $\ms{T}_e$ and invariant constraints $\ms{C}_0$ and $\ms{C}_e$, and together they form a kind $\mc{K}_0\times\mc{K}_e$ of structures consisting of the union each set of types of tensors, $\ms{T}_{\mc{K}_0\times\mc{K}_e}=\ms{T}_0\cup \ms{T}_e$, which satisfy the total constraints from both kinds $\ms{C}_{\mc{K}_0\times\mc{K}_e}=\ms{C}_0\cup \ms{C}_e$.
The tensors of the specified types therefore combine in a Cartesian product, but we keep only those that satisfy the constraints $\ms{C}_{\mc{K}_0\times\mc{K}_e}$.
The framework from \citep{Stoica2022SpaceThePreferredBasisCannotUniquelyEmergeFromTheQuantumStructure} also allows one to add more constraints $\ms{C}_{+}$ to determine the kind, as proposed in \citep{SoulasFranzmannDiBiagio2025OnTheEmergenceOfPreferredStructuresInQuantumTheoryV2}.
It's perfectly within the framework from \citep{Stoica2022SpaceThePreferredBasisCannotUniquelyEmergeFromTheQuantumStructure} to extend the set of conditions from $\ms{C}_0\cup \ms{C}_e$ to $\ms{C}_0\cup \ms{C}_e\cup\ms{C}_{+}$. 

Soulas et al. contrasted the combined structure with the extra conditions that would make it unique in the sense they call ``relational'', with the framework from \citep{Stoica2022SpaceThePreferredBasisCannotUniquelyEmergeFromTheQuantumStructure}, but it was already covered by that framework.
\thmStoica applies to any kind, including the determined one (constructions that combine structures and select a unique relation between them, as in \ref{def:soulas:relational-uniqueness}), but also much looser and more general combinations of structures.

The following quote is from \citep[section~\S3.3]{SoulasFranzmannDiBiagio2025OnTheEmergenceOfPreferredStructuresInQuantumTheoryV2}:
\begin{quote}
By applying a symmetry of $\oper{H}$ to $(\oper{H},\ket{\psi}, \mc{S}_e)$, one obtains indeed a new triple $(\oper{H},\oper{U}\ket{\psi}, \oper{U}\cdot\mc{S}_e)$, but now the input structure has been modified, thus it does not prove that a \emph{fixed} $(\oper{H},\ket{\psi})$ can not determine uniquely $\mc{S}_e$ in the sense of corollary 3.7. In short, his approach does not treat $\oper{H}$ and $\ket{\psi}$ on an equal footing: the former is considered in the absolute sense, but the latter in the relational sense.
\end{quote}

In fact, \citep{Stoica2022SpaceThePreferredBasisCannotUniquelyEmergeFromTheQuantumStructure} \emph{does} treat $\oper{H}$, $\ket{\psi}$, and $\mc{S}$ on equal footing under symmetry transformations. They transform by the same symmetry transformations, as one can see in \ref{condition-relevance} or in equation~\eqref{eq:emergent-invariant-relation}. And both $\oper{H}$ and $\ket{\psi}$ can be used as input for the desired emergent structure.
The confusion that somehow \citep{Stoica2022SpaceThePreferredBasisCannotUniquelyEmergeFromTheQuantumStructure} treats $\oper{H}$ and $\ket{\psi}$ differently may come from the fact that unitary transformations are employed to prove, based on the time-dependence of the relation between $\ket{\psi}$ and $\mc{S}$, to find physically distinct situations, as in equation~\eqref{eq:physical_relevance_invar_G}.

\section{What should we gain from an emergent structure?}
\label{s:soulas-emergence-inconsequential}

In this section I discuss Soulas et al.'s method to construct an emergent structure based on specifying its invariants, trying to see what benefits it could bring.
This discussion is independent of the analysis from section~\ref{s:TPS-soulas-nogo}, which showed that it doesn't constitute a counterexample to \thmStoica from \citep{Stoica2022SpaceThePreferredBasisCannotUniquelyEmergeFromTheQuantumStructure} and doesn't avoid the trilemma from subsection~\sref{s:TPS-soulas-trilemma}.
This section discusses the question whether this method can really construct an emergent structure, and whether it is useful in deriving predictions, retrodictions, or explanations.

\subsection{Soulas et al.'s method and the ``emergence'' of electron's mass}
\label{s:soulas-emergence-inconsequential-casimir}

Soulas et al.'s TPS is not uniquely determined by $\oper{H}$ and $\ket{\psi}$ only, but also by many other constants $s_{R,k}$ encoded in the function $s$ from equation~\eqref{eq:sRi}, correlated to fix a finite number of parameters. But they are purely relational constraints for the TPS, they are invariant to a change of basis of $\hilbert$, because they are scalars obtained by combining only $\ket{\psi}$, $\oper{H}$, and the candidate TPS. Moreover, Soulas et al. combine them all in a single condition, the invariant property $P_L$ in their article, that can be stated very simply as
\begin{equation}
\label{eq:unified-invariant-property}
P_L=\texttt{True},
\end{equation}
or as ``the von Neumann entropies equal the values of $s_{R,k}$''. This may remind us of Feynman's unified equation of the entire physics, ``unwordliness=0'' \citep[ch. 25 page 10]{Feynman1964LecturesOnPhysicsVolII}.

If choosing many invariants, arbitrarily or conveniently to get an already decided TPS, is not to be understood as emergence, maybe the construction done by Soulas et al. doesn't really take place within \HSF, but in a modification of \HSF.
Using the values $s_{R,k}$ as an additional input modifies the original \HSF thesis into
\setdefCustomtag{HSF$(s_{R,k})$}
\begin{defCustom}[\HSF $+$ von Neumann entropies]
\label{def:HSFs}
Everything about the world is encoded in
\begin{equation}
\label{eq:HSFs}
\(\hilbert,\oper{H},\ket{\psi}{\color{red},s_{R,k}}\),
\end{equation}
where $s_{R,k}$ are as in equation~\eqref{eq:sRi}.
\end{defCustom}

\ref{def:HSFs} is equivalent to the modification of \HSF that includes the TPS, \ref{def:HSF_TPS}.

We can make an analogy between Soulas et al.'s idea to specify the TPS by specifying its entanglement invariants and the classification of elementary particles.

The orbits of the Lorentz group's action on the four-dimensional Minkowski space are lightcones and hyperboloids. This is similar to the fact that the orbits of the action of the rotation group are concentric spheres classified by the radius. The difference comes from the fact that the Minkowski metric is indefinite. This gives a nice representation on the four-dimensional space of four-momenta, and a classification of particles obtained by classifying the orbits. Positive mass particles correspond to hyperboloids in the future lightcone, zero mass particles to lightcones, and the particle's mass itself is an invariant classifying these orbits.

Similarly, the action of the full Poincar\'e group on the projective Hilbert space, action represented by unitary and anti-unitary operators, can be decomposed in irreducible representations. These are infinite-dimensional, and are classified as well by the mass of the particle, but also by its spin \citep{Wigner1931GruppentheorieUndIhreAnwendungAufDieQuantenMechanikDerAtomspektren,Wigner1939OnUnitaryRepresentationsOfTheInhomogeneousLorentzGroup,Bargmann1954OnUnitaryRayRepresentationsOfContinuousGroups}. The mass and the spin are the Casimir invariants of the representations. As a side note, this is how to break the symmetry of the Hilbert space in quantum field theory, to make space appear as needed by \qtref{QT:OBS} \citep{Stoica2024WhyTheWavefunctionAlreadyIsAnObjectOnSpace}.

The same insight is used in the standard model of elementary particles, where the Poincar\'e group acts on spacetime while other groups act on the fibers of a bundle, giving the gauge-theoretical description of interactions and add more invariants, \eg charges, to classify the elementary particles \citep{Weinberg2005QuantumTheoryOfFields}.

Now, recall that the construction from Soulas et al. is based on fixing the invariants classifying the orbits (\ie of the TPSs). Fixing the von Neumann entropies to the values $s_{R,k}$ selects a unique TPS.

But, if Soulas et al.'s construction is justified, wouldn't it be equally justified to claim that the electron's mass emerges uniquely, by fixing by hand the invariants of the irreducible projective representations of the Poincar\'e group?
Of course, nobody claims this in the case of the electron, because fixing the Casimir invariants amounts to fixing the electron's mass and spin by hand.
But how is this different when we classify the TPSs?
\citep{SoulasFranzmannDiBiagio2025OnTheEmergenceOfPreferredStructuresInQuantumTheoryV2} doesn't discuss this.

Then, what is the difference between Soulas et al.'s construction and the direct specification of the electron's mass?
In both cases, one selects an orbit by fixing invariants.
The invariants are exactly the data that define the structure.
Therefore the ``emergence'' is merely orbit selection.

A difference is that, in the case of the electron, it is immediately evident that specifying the invariant is the same as specifying the mass. This is why nobody claims that the electron's mass emerges this way, because it would evidently amount to specifying it directly. Programs that attempt to show that the electron's mass emerges, try to do this from something else, for example from the total energy of the electromagnetic field of the electron, or by integrating the Einstein tensor, in the case of electromagnetic fields in general relativity.

\begin{question}
\label{s:TPS-no-constraints-really}
Is there a sense in which we can call the TPS from \citep{SoulasFranzmannDiBiagio2025OnTheEmergenceOfPreferredStructuresInQuantumTheoryV2} emergent, which at the same time doesn't imply that fixing the mass of the electron by hand is also emergence?
Does their construction constrain the TPS based on $\oper{H}$ and $\ket{\psi}$, or it just specifies any TPS we could wish for?
\end{question}

Using the Soulas et al. construction, we can clearly specify any TPS we may wish, without any restriction. For this, we just start with the target TPS, collect the von Neumann entropies of the subsystems of $\ket{\psi}$, and we get the needed $s_{R,k}$.

But, in my opinion, when we say that a TPS $\mc{T}$ emerges from $\oper{H}$ and $\ket{\psi}$, there should at least be a limitation of the possible resulting $\mc{T}$, and this limitation should be different if $\oper{H}$ or $\ket{\psi}$ are different. Otherwise in what sense can we say that $\mc{T}$ emerges from $\oper{H}$ and $\ket{\psi}$?

\subsection{What should we expect to gain from an emergent structure?}
\label{s:soulas-emergence-inconsequential-gain}

If someone would come up with a way to compute the electron's mass from other quantities, we can test this empirically. We can compare the experimental values with the ones predicted by the theory, and decide whether to keep or to reject the theory. 

I think an emergent structure must do some job, it must unify something, make new prediction, explain previously unexplained relations, explain some phenomena reductively in terms of other phenomena. Isn't this the whole purpose of \HSF? 

For example, when Newton unified the law of gravity describing apples falling with Kepler's laws, this brought something new, a better explanation.
When Einstein unified space and time, his theory also unified momentum and energy in a four-momentum, and later it gave us the classification of elementary particles. When Maxwell's equations were written in a Poincar\'e invariant form, the 3D electric vector field and magnetic pseudo-vector field were unified in a 4D tensor.
When electrodynamics was formulated as a gauge theory, this could be used for other interactions.
Every time, a new explanation was born. This explanatory power could be extended to help us make other discoveries, and to predict the values of some properties or constants in terms of other properties or constants.
This is what I consider emergence to be, if you ask me. But fixing the electron's mass by fiat to any value that comes from the measurements contributes nothing, and I don't consider it emergence.

This aspect related to Question~\ref{s:TPS-no-constraints-really} can be formulated as
\begin{question}
\label{q:TPS-soulas-inconsequential}
Does the construction from \citep{SoulasFranzmannDiBiagio2025OnTheEmergenceOfPreferredStructuresInQuantumTheoryV2} lead to the emergence of new predictions, retrodictions, or explanations of some phenomena?
\end{question}

It is natural to expect that the predictions or explanations resulting from \HSF to be exactly its central claims, the emergence of preferred structures that we notice in the world.
\HSF shouldn't simply try to fix the TPS we already noticed in the world, but to exclude other TPSs.
It may, for example, give an alternative explanation of the fermionic and bosonic statistics.

\subsection{Example of emergent structure: Cotler et al.'s Hamiltonian}
\label{s:cotler-emergence-example}

I will give an example in which it makes sense to say that there is emergence, including in the predictive or explanatory senses:

The construction from \citep{CotlerEtAl2019LocalityFromSpectrum} severely restricts both the Hamiltonian and the TPS (even though the TPS is not unique as it was widely believed).
It restricts $\oper{H}$ because almost all possible $\oper{H}$ don't have a $k$-local form. It restricts the TPS because, once $\oper{H}$ found, it doesn't have a $k$-local form in all possible TPSs.
It is conceivable that these restrictions result in specific empirical predictions, at least in principle.

\subsection{What can we learn from Soulas et al.'s TPS?}
\label{s:soulas-TPS-lesson}

Because \ref{def:HSFs} works for any TPS, it doesn't explain more than \ref{def:HSF_TPS}, except for a more complicated and indirect specification of the TPS. In other words, it does not challenge the independence of \qtref{QT:TPS} from \qtref{QT:HSF}.

Of course, not every contribution to physics has to make predictions or to explain something. The construction from \citep{SoulasFranzmannDiBiagio2025OnTheEmergenceOfPreferredStructuresInQuantumTheoryV2} seems intended to be a counterexample to \citep{Stoica2022SpaceThePreferredBasisCannotUniquelyEmergeFromTheQuantumStructure}, and if it could, this would be more than sufficient. 
My point is not that a counterexample \emph{should} make predictions, but that \emph{if} it really proves an emergent phenomenon, \emph{then} it is natural to expect that this emergence comes with predictions or explanations.
So maybe my Question~\ref{q:TPS-soulas-inconsequential} is just a comment on the meaning of the notion of ``emergent preferred TPS'' or ``emergent preferred structure'' in general.

However, we have seen that Soulas et al.'s construction, once the invariants $s_{R,k}$ fixed, does make predictions, but not in the form of an instantaneous relation between $\oper{H}$, $\ket{\psi}$, and the TPS, which is what it may be natural to expect from a proof of emergence. Its predictions appear in another place: when we apply it to the dynamics of quantum theory. And these predictions allow the falsification of the claim that the TPS emerges, but the resulting TPS is not refuted, since this one can be any TPS.

I personally wouldn't consider fixing the invariants by hand to select a unique TPS to be ``unique emergence'' or even emergence at all.
But maybe that's just my opinion.
Nevertheless, we've seen that the no-go results from \citep{Stoica2022SpaceThePreferredBasisCannotUniquelyEmergeFromTheQuantumStructure}, in particular \thmStoica, cover any such inclusive but inconsequential notion of emergence as well.

\section{On the ``relational philosophy'' of Soulas et al.}
\label{s:soulas-relational-philosophy}

A major claim in Soulas et al.'s criticism of \citep{Stoica2022SpaceThePreferredBasisCannotUniquelyEmergeFromTheQuantumStructure} comes from their perception that the latter doesn't abide by ``relational philosophy'', misattributing it an ``absolute perspective''.

It is true that \thmStoica refutes Soulas et al.'s relational philosophy as they stated it, but it does so because it contradicts empirical evidence, and not because of a hypothetical adherence to an absolute perspective.
I will explain more about this in this section.
We will also see that \citep{Stoica2022SpaceThePreferredBasisCannotUniquelyEmergeFromTheQuantumStructure}, and not Soulas et al.'s construction, really adopts the correct relational position resulting from the symmetry group of quantum theory.

\subsection{Soulas et al.'s ``relational philosophy'' \vs physical evidence}
\label{s:soulas-relational-philosophy-refutes-itself}

Can a no-go result say something about a philosophical position?

\citep[section~\S3]{SoulasFranzmannDiBiagio2025OnTheEmergenceOfPreferredStructuresInQuantumTheoryV2} motivate their position by invoking a ``relational philosophy'':
\begin{quote}
We argue that the relevant notion of uniqueness in physics is relational, meaning that a preferred emergent structure should be unique up to a global unitary equivalence on the whole set of structures involved. 
Indeed, \emph{two sets of structures in a Hilbert space related by a global unitary will yield exactly the same theoretical predictions}, hence they depict one and the same physical theory.
\end{quote}

Soulas et al. consider that HSF as stated in \HSF abides by \citeauthor{SoulasFranzmannDiBiagio2025OnTheEmergenceOfPreferredStructuresInQuantumTheoryV2}'s relational philosophy.
The triple $\(\hilbert,\oper{U}\oper{H}\oper{U}^\dagger,\oper{U}\ket{\psi}\)$ depicts one and the same physical reality as the triple $\(\hilbert,\oper{H},\ket{\psi}\)$.
And, since any structure $\mc{S}$ derived uniquely by invariant methods only from $\oper{H}$ and $\ket{\psi}$, in particular TPS $\mc{T}$, should be understood to be unique up to global unitaries, one gets that
$\(\hilbert,\oper{U}\oper{H}\oper{U}^\dagger,\oper{U}\ket{\psi},\oper{U}\cdot\mc{S}\)$ depicts one and the same physical reality as $\(\hilbert,\oper{H},\ket{\psi},\mc{S}\)$.
But we've seen that this implies immediately that this should be true even if we choose $\oper{U}=e^{-\frac{i}{\hbar}\oper{H}t}$.
So relational philosophy implies that the world now is exactly like the world at any other time, no matter how distant in the past or in the future. More details can be found in \citep{Stoica2026NoChangeInHilbertSpaceFundamentalism}.

And we've seen in detail in section~\sref{s:TPS-soulas-nogo}, on the very TPS proposed by Soulas et al., that in \HSF we get the exact same description of the world at all times, regardless of the route we adopt trying to fix the emergent structure.

Simply put, Soulas et al.'s relational philosophy can't distinguish the changes in time.
Moreover, it can't distinguish the actual state of the world $\ket{\psi(t)}$ from alternative possibilities $\oper{U}\ket{\psi(t)}$, where $\oper{H}=\oper{U}\oper{H}\oper{U}^\dagger$.
This means that nature herself doesn't abide by Soulas et al.'s relational philosophy.
But does it contradict relational philosophy in general?

\subsection{The correct relational perspective is relative to the fundamental structures}
\label{s:soulas-relational-philosophy-relative}

Relationalism manifests differently in different theories, depending on their fundamental structures. Empirical evidence decides which is the right one.

Let us start with a simpler case.
What is the difference between affine geometry and Euclidean geometry? Affine geometry is about lines and their incidence relations, \ie whether they intersect, coincide, or are parallel. Affine geometry doesn't know about angles or distances, or even about similarity.
For this reason, its symmetry group consist of general linear transformations and translations. In plane, the affine group has $6$ dimensions, while the Euclidean group of isometries has $3$ dimensions.
Affine symmetry transformations preserve the lines and their intersection points and parallelism, and also collinearity relations between points. Euclidean symmetry transformations preserve all of these and angles and distances, \ie the metric.

For this reason, from the point of view of Euclidean geometry, a general affine transformation deforms a configuration of geometric shapes. Any triangle can be mapped by an affine transformation into any other triangle, even if the two triangles are not congruent according to Euclidean geometry.
Isometries are affine transformations, but most affine transformations are not isometries.

The ``correct relational philosophy'' consists of applying the appropriate symmetry group. We apply the affine group to affine geometry, and the Euclidean group to Euclidean geometry.
The invariant properties of affine geometry are also invariant properties of Euclidean geometry, but not the other way around.
So to move from affine geometry to Euclidean geometry we need to reduce the symmetry group, for example by requiring the symmetry transformations to preserve the metric. Alternatively, we can add more postulates, the congruence axioms.

If we compare these two geometries with the empirical evidence, to the extent that our physical space can be approximated by one of them, this is the Euclidean geometry.
Affine geometry would proclaim the triangle (Sirius, Polaris, Betelgeuse) as equivalent with the Bermuda triangle (Florida, Bermuda, Puerto Rico) and with the printed symbol $\triangle$, while Euclidean geometry sees them as very different. And not because Euclid adopted the wrong philosophy or an ``absolute perspective''.

\begin{remark}
\label{rem:affine-to-Euclid}
One may imagine that we can ``derive'' the Euclidean metric within affine geometry by fixing some invariants. For example, we can impose the condition that the six distances between the Moon, the Earth, the Sun, and Jupiter are equal to six predefined values $s_{j,k}$.
Once we fix these values, provided that they satisfy strictly certain inequalities, we can indeed reconstruct an Euclidean geometry of space. But for different choices of the values $s_{j,k}$, we will get different metrics, different Euclidean geometries associated to the same space. And these distances change in time, which comes with additional problems: should we make $s_{j,k}$ time-dependent in the exact way needed to compensate for these changes, or should we assume an absolute moment of time when the Euclidean metric ``emerged uniquely'' from affine geometry, and fix that metric once and for all?
\end{remark}

This is similar to what we've seen happens when we compare \HSF with quantum theory.
\HSF may proclaim that it can derive \qtref{QT:TPS} and \qtref{QT:OBS} from \qtref{QT:HSF}, but as long as the ``\qtref{QT:HSF} geometry'' and the ``\qtref{QT:HSF}$+$\ref{QT:TPS}$+$\ref{QT:OBS} geometry'' have different symmetry groups, they are not the same.
Just like we can't derive the congruence axioms from affine geometry only.
And just like affine geometry sees all triangles as the same triangle, \HSF can't distinguish the state of the world at different times. Quantum theory distinguishes them, and therefore nature voted for quantum theory.

The symmetry group of quantum theory and of the world it describes is much smaller than the unitary group of the Hilbert space. It includes the group of spacetime isometries (if we ignore gravity) and the gauge groups of the standard model, but not the full unitary group.
Unitary transformations outside of these symmetries take states into physically distinct states. To be applied correctly, ``relational philosophy'' should be applied to this group, not to the full unitary group.
And this is why \qtref{QT:TPS} and \qtref{QT:OBS} are essential, and can't be eliminated, just like we can't eliminate the congruence axioms in Euclidean geometry.

Does this contradict relational philosophy? Only the version that claims that the symmetry group from physics is the full unitary group.
The correct relational position is relative to the fundamental structures of \qtref{QT:HSF}$+$\ref{QT:TPS}$+$\ref{QT:OBS}.

\subsection{Embracing the ambiguity: many-many-worlds}
\label{s:soulas-relational-philosophy-embrace-ambiguity}

We've seen that the relational philosophy of Soulas et al. was applied incorrectly, and \HSF is quite ambiguous, being unable to distinguish changes in time or alternative states of the world.

But what if we embrace this ambiguity? What if it's perfectly fine, or even better, to describe the entire history and even alternative histories with the same triple $\(\hilbert,\oper{H},\ket{\psi}\)$ and whatever structures emerge uniquely from it?

What if all worlds encoded in the same state vector are equally real, and so there are ``parallel many worlds'' within the same state, independently of those proposed by Everett?

This ``biting the bullet'' attitude was already discussed in \citep[section~\S$7.2$]{Stoica2022SpaceThePreferredBasisCannotUniquelyEmergeFromTheQuantumStructure}. The ambiguity of an emergent pointer basis and the idea of parallel many worlds were partially rediscovered later by \citep{AdilAlbrecht2024SearchForClassicalSubsystemsInQuantumWorlds}, with the help of the computer.

However, even the overly inclusive version of \HSF, the one that accepts that the triple $\(\hilbert,\oper{H},\ket{\psi}\)$ describes at once many different worlds, makes empirical predictions: it predicts different correlations between events. In particular, it implies that there are no correlations between the properties of external objects and the records of these properties from our brains \citep{Stoica2023AreObserversReducibleToStructures}, directly contradicting the very fact that we can know the values of properties of the external world.

\section{``Emergent World'' \vs ``Emergent Law'' programs}
\label{s:world-v-law}

I think that a source of problems for the \HSF program is its aim to be a theory of everything about the world. This invites the temptation to use some results beyond their scope to build entire theories of emergence.
But since, as we have seen already in subsection~\sref{s:intro:HSF-no}, \HSF is oblivious to temporal changes, a good move seems to be to reduce the claim that \HSF describes the world to the weaker claim that \HSF can at least deal with the unchanging things. Laws don't change, so maybe \HSF is only about laws.

\subsection{Example of a result within ``HSF about laws''}
\label{s:world-v-law-cotleretal}

\citep{CotlerEtAl2019LocalityFromSpectrum}'s construction was believed to give us for free, from the Hamiltonian only, a preferred TPS. Their TPS was then used to determine a network of points forming the 3D space \citep{CarrollSingh2019MadDogEverettianism,Carroll2021RealityAsAVectorInHilbertSpace}.
The construction from \citep{CotlerEtAl2019LocalityFromSpectrum} is significant, but it can't be used in this way, because it's not unique in the needed sense (see Remark~\ref{rem:cotleretal-TPS}).

\citep{CotlerEtAl2019LocalityFromSpectrum} define a notion of equivalence of TPS w.r.t. the Hamiltonian, but this is not the same as uniqueness, not even up to physically unobservable differences. 
Two distinct TPSs $\mc{T}$ and $\mc{T'}$ that are globally equivalent would exhibit physically distinct entanglement entropies for the same state vector $\ket{\psi}$. I pointed this out in 2021, in reply to an email received from Daniel Ranard, coauthor of \citep{CotlerEtAl2019LocalityFromSpectrum}, after I uploaded the preprint of \citep{Stoica2022SpaceThePreferredBasisCannotUniquelyEmergeFromTheQuantumStructure}.

A few years later, on July 3rd 2025, Daniel emailed Antoine Soulas about his first preprint of another paper \citep{Soulas2025DisentanglingTensorProductStructures}, including me to the email. In my first reply to Daniel's email I wrote on July 3rd
\begin{quote}
can ``a unique TPS'' as in this definition [see \sref{s:TPS-basics}] be the same as ``a unique TPS up to an equivalence'' as in the definition you proposed to use [\emph{Equivalence of TPS} from \citep{CotlerEtAl2019LocalityFromSpectrum}]?
I think we all agree that they are not the same.
\end{quote}

Daniel agreed that ``up to an equivalence'' and ``unique'' are different, so the first should not be used as meaning the second, and we discussed what options we have to clarify this.

On July 7th, a few days after my reply that included the clarification and Daniel agreed, Antoine sent his first reply to Daniel's email, writing that he agreed too and they missed this distinction. If there was any debate about this between Daniel and myself, we resolved it long before its resolution was announced by \citep{SoulasFranzmannDiBiagio2025OnTheEmergenceOfPreferredStructuresInQuantumTheoryV2}.

But the more important point I want to make here is that the construction from \citep{CotlerEtAl2019LocalityFromSpectrum}, while believed to obtain a unique TPS, \emph{is more interesting for what it does than for what was misunderstood to achieve}:
\begin{remark}
\label{rem:locality-from-the-spectrum}
I think \citep{CotlerEtAl2019LocalityFromSpectrum}'s results should \emph{not} be understood as being about the emergence of a unique TPS (and clearly is not unique), but about the existence of \emph{a unique local form of the Hamiltonian}. The TPS is incidental, its only role is to exhibit the local form of the Hamiltonian, when expressed as linear combinations of tensor products between local operators on the factors of the TPS.
This exhibits the interactions between subsystems.
\end{remark}

And indeed, \citep{CotlerEtAl2019LocalityFromSpectrum} discusses extensively about \emph{local descriptions} of the Hamiltonian, and I think the title itself, ``Locality from the spectrum'', suggests that it should be understood as about the form of the Hamiltonian, rather than about the TPS.
It makes no sense to say of the TPS that it is local, but it makes perfect sense to say this about the Hamiltonian.
So the result from \citep{CotlerEtAl2019LocalityFromSpectrum} is about the law, about decoding the interactions encoded in the Hamiltonian.

\subsection{``\HSF about ontology'' \vs ``\HSF about law''}
\label{s:world-v-law-ontology-vs-law}

In classical mechanics, the Hamiltonian function $\ms{H}(\q,\p)$ specifies the theory, but it has nothing to say about particular structural configurations, which depend of the initial conditions. Moreover, canonical transformations preserve the form of Hamilton's equations, but not necessarily the form of the Hamiltonian $\ms{H}(\q,\p)$ as a function of $\p$ and $\q$, in particular of the subsystems (particles or fields).
But there are canonical transformations that also preserve the form of the Hamiltonian $\ms{H}(\q,\p)$. When used as active transformations, they ``rewrite'' the history of the system. When used as passive transformations, they change the generalized coordinates, but the positions and momenta used to describe the same system look as if they describe an alternative world, in which the particles are arranged completely differently in space and have different momenta.
Assuming for a moment that there is a classical description of the phases of water for example, a canonical transformation would be able to describe water in the way it describes ice or vapor, and the other ways around.
And a pattern of water waves or ice crystals would look like another such pattern or like a cloud.
But nobody worries about this in classical physics, because it is assumed that the classical particles and classical fields in space are the true ontology of the classical world.

The quantum analog of the classical Hamiltonian function is the Hamiltonian operator $\oper{H}$. The quantum analog of classical canonical transformations are unitary transformations. The dependence of $\ms{H}$ of the classical observables $\p$ and $\q$ becomes, after quantization, a dependence of $\oper{H}$ of the quantum observables $\widehat{\p}$ and $\widehat{\q}$.

Hilbert space fundamentalism claims that we can get rid of these observables, and keep only the Hamiltonian and a unit vector \emph{as the ontology of the world}.

The \HSF program is structural, but also ontological \citep{Carroll2021RealityAsAVectorInHilbertSpace}:
\begin{quote}
Here I want to argue for the plausibility of an extreme position among these alternatives, that the fundamental ontology of the world is completely and exactly represented by a vector in an abstract Hilbert space, evolving in time according to unitary Schr\"odinger dynamics. Everything else, from particles and fields to space itself, is rightly thought of as emergent from that austere set of ingredients. This approach has been called ``Mad-Dog Everettianism'' \citep{CarrollSingh2019MadDogEverettianism} although ``Hilbert Space Fundamentalism'' would be equally accurate.
\end{quote}
The quote from \citep{CarrollSingh2019MadDogEverettianism} reproduced in the \nameref{s:intro} makes similar claims: ``The simplest quantum ontology is...''. Everything about the world should be completely ``decodeable'' from just the triple $\(\hilbert,\oper{H},\ket{\psi(t)}\)$.
And indeed, \HSF tries to recover a description of the world, for example of the 3D space and its curved metric, or even \citep{Carroll2021RealityAsAVectorInHilbertSpace}:
\begin{quote}
The structure of our observed world, including space and fields living within it, should arise as a higher-level emergent description.
\end{quote}

We've seen that this doesn't work, just like a classical triple consisting of a phase space, a Hamiltonian function $\ms{H}$, and a state $(\q,\p)$ can describe infinitely many alternative worlds under a canonical transformation. The theory would not be able to distinguish different patterns of water waves or ice crystals or cloud shapes, or, in fact, even entirely different worlds. And to fix this ambiguity, we need to assume that $(\q,\p)$ has attached a unique physical meaning in terms of particles and their positions and momenta. This is the classical analog of \qtref{QT:OBS} from the \nameref{s:intro}.

\begin{remark}
\label{rem:main-lesson}
The main lesson from \citep{Stoica2022SpaceThePreferredBasisCannotUniquelyEmergeFromTheQuantumStructure} is that \HSF can't work, because the same triple $\(\hilbert,\oper{H},\ket{\psi(t)}\)$ can describe infinitely many different physical worlds.
\qtref{QT:OBS} is necessary.
\end{remark}

But can a weaker version of the \HSF program survive?
Perhaps, if it drops the claim that it describes the world, and focus only of certain relational aspects of the laws.
Maybe this version:
\setdefCustomtag{HSF-law}
\begin{defCustom}[law-oriented \HSF]
\label{def:HSF-law}
Two triples $\(\hilbert,\oper{H},\ket{\psi}\)$ and $\(\hilbert',\oper{H'},\ket{\psi'}\)$ depict one and the same physical law if and only if they are isomorphic as in equation~\eqref{eq:HSF-iso}.
\end{defCustom}

As the example about \citep{CotlerEtAl2019LocalityFromSpectrum} demonstrates, there are interesting things to learn about the law, in the form of the Hamiltonian, from its spectrum alone. But as seen in \citep{Stoica2022SpaceThePreferredBasisCannotUniquelyEmergeFromTheQuantumStructure} and in this paper, this shouldn't be used to make claims about the particular configuration of the world, let alone about the ontology.

Then, the original, strong version of \HSF can be called ``world-oriented \HSF'' or ``ontic \HSF'', and the weaker version ``law-oriented \HSF'' or ``nomic \HSF''.

\subsection{Limits of ``\HSF about law''}
\label{s:world-v-law-limits}

In a recent talk \citep{Ranard2026LocalityFromTheSpectrumTalk}, Daniel said that \HSF seems ``ugly'', because, unlike the Lagrangian of the standard model, which can fit on a T-shirt, ``the spectrum is a \emph{huge} amount of data'' that wouldn't fit on a T-shirt.
But who knows? Maybe we will find out that the spectrum of the Hamiltonian is a simple sequence like $\omega_j=\sqrt{j}$, or $\omega_j=j$, or the prime numbers.

\citep{Stoica2024EmpiricalAdequacyOfTimeOperatorCC2HamiltonianGeneratingTranslations} gives infinitely many examples of physically different theories with the same Hamiltonian
\begin{equation}
\label{eq:H-translation}
\oper{H}=-i\hbar\pdv{\tau}.
\end{equation}

The examples include any quantum world containing a clock, a free massless fermion, or a standard ideal quantum measurement, and the quantum representation \citep{Koopman1931HamiltonianSystemsAndTransformationInHilbertSpace,vonNeumann1932KoopmanMethod} of any deterministic time-reversible classical dynamical system, and any quantum world that cannot return to a past state \citep{Stoica2022ProblemOfIrreversibleChangeInQuantumMechanics}. Also the formalism that recovers time in timeless quantum gravity developed by \citet{PageWootters1983EvolutionWithoutEvolution} turns out to share this Hamiltonian. The Hamiltonians of these quantum worlds look completely different because the physical meaning of their observables and the TPSs are different, but when expressed in a particular way it's just~\eqref{eq:H-translation}.
If our world is like one of these, the spectrum of the Hamiltonian is simply
\begin{equation}
\label{eq:H-spectrum-translation}
	\includegraphics[width=0.25\textwidth]{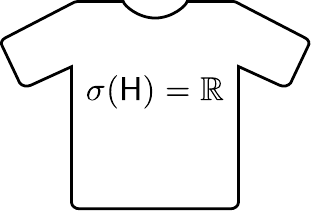}
\end{equation}

But, given how many theories with physically different laws share the Hamiltonian~\eqref{eq:H-translation}, the spectrum alone, and therefore \ref{def:HSF-law} too, is insufficient in this case.
As shown in \citep{Stoica2022ProblemOfIrreversibleChangeInQuantumMechanics}, no matter how complicated the Hamiltonian of the world is, the existence of a particle in linear motion, an ideal measurement device, or an ideal clock, smears its Hamiltonian, turning it into~\eqref{eq:H-translation}.
The interactions encoded in the relation between the Hamiltonian and the TPS are an essential part of the law, and perhaps they constitute most of it, but there are many ways to decode interactions from the Hamiltonian alone. Therefore, it seems that even \ref{def:HSF-law} has a very limited scope.

\section{Conclusion}
\label{s:conclusion}

A discussion of the emergence of preferred structures in quantum theory has to disentangle many aspects that make understanding difficult, despite the existence of simple arguments.
But Soulas et al.'s construction offered the perfect example that illustrates the obstructions for the emergence of preferred structures, as well as common hidden assumptions behind this idea.

We examined the TPS construction from \citep{SoulasFranzmannDiBiagio2025OnTheEmergenceOfPreferredStructuresInQuantumTheoryV2}. 
Despite the many choices, once these made, the construction method is invariant at every step, if all that matters is the triple $(\hilbert,\oper{H},\ket{\psi})$, rather than a dynamical system $(\hilbert,\oper{H},\ket{\psi(t)})$.
The construction relies on many arbitrary choices for the invariants~$s_{R,k}$.
An easy to spot issue is that, by this, the construction can accommodate any TPS, so it doesn't give a reason why a TPS is better than any other TPS.

But the central question is how does this work as a quantum theory. We've seen that once we take dynamics into account, the invariance is challenged, and trying to restore it conflicts with observable evidence, as shown by \thmStoica from \citep{Stoica2022SpaceThePreferredBasisCannotUniquelyEmergeFromTheQuantumStructure}.
For this reason, the TPS construction from \citep{SoulasFranzmannDiBiagio2025OnTheEmergenceOfPreferredStructuresInQuantumTheoryV2}, even if it was intended as a counterexample to \citep{Stoica2022SpaceThePreferredBasisCannotUniquelyEmergeFromTheQuantumStructure}, turns out to be a perfect illustration of those no-go results.

We also had the opportunity to develop the correct meaning of relational perspective and to pick up what survives from \HSF, proposing a weaker version of \HSF that has chances to succeed by focusing on limited aspects of the physical law rather than aiming to be a theory of everything or provide a complete description of reality.

\addcontentsline{toc}{section}{\refname}


\end{document}